\documentclass[acmsmall, screen]{acmart}

\usepackage{amsmath}
\usepackage{amsthm}
\usepackage{amsfonts}

\usepackage{thm-restate}

\usepackage{todonotes}

\usepackage{tikz-cd}

\newtheorem{remark}{Remark}[section]

\usepackage{xcolor}
\usepackage{tikz, tikz-qtree, tikz-cd}
\usepackage{graphicx}
\usetikzlibrary{arrows,calc,patterns,positioning,shapes}
\usetikzlibrary{decorations.pathmorphing}
\tikzset{
modal/.style={>=stealth’,shorten >=1pt,shorten <=1pt,auto,
node distance=0.6cm},
world/.style={circle,draw,minimum size=0.6cm},
point/.style={circle,draw,fill=black,inner sep=0.5mm},
reflexive/.style={->,in=120,out=60,loop,looseness=#1},
reflexive/.default={5},
reflexive point/.style={->,in=135,out=45,loop,looseness=#1},
reflexive point/.default={25},
}

\newcommand{\Prop}{\text{Prop}}
\newcommand{\simu}{\;\underline{\to}\;}

\newcommand{\bisim}{\;\underline{\leftrightarrow}\;}
\newcommand{\negbisim}{\;\underline{\not\leftrightarrow}\;}

\newcommand{\mono}{\mathcal{L}_{\Box,\Diamond,\wedge,\vee}}

\newcommand{\posexcon}{\mathcal{L}_{\Diamond,\wedge,\top}}
\newcommand{\posex}{\mathcal{L}_{\Diamond,\wedge,\vee,\top}}
\newcommand{\pos}{\mathcal{L}_{\Box,\Diamond,\wedge,\vee,\top,\bot}}

\newcommand{\emptyloop}{\circlearrowright_{\emptyset}}
\newcommand{\fullloop}{\circlearrowright_{\mathrm{Prop}}}

\newcommand{\upclosure}{\mathop{\uparrow\!}}
\newcommand{\downclosure}{\mathop{\downarrow\!}}

\usepackage{todonotes}

\newcommand{\revision}[1]{{\color{red}#1}}
\newcommand{\rrevision}[1]{{\color{blue}#1}}
\renewcommand{\revision}[1]{#1}
\renewcommand{\rrevision}[1]{#1}

\AtBeginDocument{%
  \providecommand\BibTeX{{%
    \normalfont B\kern-0.5em{\scshape i\kern-0.25em b}\kern-0.8em\TeX}}}

\setcopyright{acmcopyright}
\copyrightyear{2022}
\acmYear{2022}
\acmDOI{XXXXXXX.XXXXXXX}

\begin{document}

\title{Characterising Modal Formulas with Examples}


\author{Balder ten Cate}
\orcid{0000-0002-2538-5846}
\affiliation{%
  \institution{ILLC, University of Amsterdam}
  \streetaddress{Science Park 107}
  \city{Amsterdam}
  \country{The Netherlands}}
\email{b.d.tencate@uva.nl}

\author{Raoul Koudijs}
\orcid{0000-0002-9000-4675}
\affiliation{%
  \institution{University of Bergen}
  \streetaddress{Thormøhlens Gate 55, 5008 Bergen, Norway}
  \city{Bergen}
  \country{Norway}}
\email{raoulxluoar@gmail.com}

\begin{abstract}
We study the existence of finite characterisations for modal formulas. A finite characterisation of a modal formula $\varphi$ is a finite collection of positive and negative examples that distinguishes $\varphi$ from every other, non-equivalent modal formula, where an example is a finite pointed Kripke structure. This definition can be restricted to specific frame classes and to fragments of the  modal language: a modal fragment $\mathcal{L}$ admits finite characterisations with respect to a frame class $\mathcal{F}$ if every formula $\varphi\in\mathcal{L}$ has a finite characterisation with respect to $\mathcal{L}$ consisting of examples that are based on frames in $\mathcal{F}$. Finite characterisations are useful for illustration, interactive specification, and debugging of formal specifications, and their existence is a precondition for exact learnability with membership queries. We show that the full modal language admits finite characterisations with respect to a frame class $\mathcal{F}$ only when the modal logic of $\mathcal{F}$ is locally tabular. We then study which modal fragments, freely generated by some set of connectives, admit finite characterisations. Our main result is that the positive modal language without the truth-constants $\top$ and $\bot$ admits finite characterisations w.r.t.~the class of all frames. This result is essentially optimal: finite characterizability fails when the language is extended with the truth constant $\top$ or $\bot$ or with all but very limited forms of negation.
\end{abstract}

\terms{Languages, Design, Theory}



\maketitle

\section{Introduction}

By the standard Kripke semantics of modal logic, we can associate to 
every modal formula $\varphi$ the (possibly infinite) class of pointed models that satisfy it, and, by complementation, the set of pointed models that do not satisfy it. We will call the former \emph{positive examples} of $\varphi$, and
we will call the latter \emph{negative examples} of $\varphi$.
We study the question whether it is possible to \textit{characterise} a modal formula $\varphi$, by a finite set of positive and negative examples.
That is, we study the question whether there is a finite collection of positive and negative examples, such that $\varphi$ is the only formula (up to logical equivalence), that fits these examples. Furthermore, ideally, we would like these examples to be finite structures. If such a uniquely characterizing finite set of finite examples exists, we call it a \textit{finite characterisation} of $\varphi$.


Generating a uniquely characterizing set of  examples for a given logical specification can be useful for illustration, interactive specification, and debugging purposes (e.g.,~\cite{MannilaR86} for relational database queries, \cite{AlexeCKT2011} for schema mappings, and \cite{StaworkoW15} for XML queries). The exhaustive nature of the examples is useful in these settings, as, intuitively, they essentially display all the `ways' in which the specification can be satisfied or falsified.
The existence of finite characterisations is also important as a precondition for learnability in the model of \textit{exact learning with membership queries}~\cite{AngluinQueriesConcepts}: (efficient) exact learnability with membership queries is only possible when the concept
class in question admits a finite (polynomial-size) unique characterizations. Conversely, an
effective procedure for constructing finite characterizations gives rise to an effective
(but not necessarily efficient) exact learning algorithm with membership queries.
Exact learnability is an active topic of research in knowledge representation, where, in particular, the learnability of concept expressions in various description logics is studied,
corresponding to fragments of modal logic~\cite{FunkJL2022,Funk2021ELr,tencatedalmauCQ,LTLWolter,PatrikThesis}.

We begin by observing that \emph{no modal formula has a finite characterization with respect to the full modal language}, if we consider the modal logic $\textbf{K}$ (i.e., the modal logic of the class of all frames). This initial negative result  leads us to 
explore two questions: 

\begin{enumerate}
    \item For which frame classes $K$ does it hold that the modal logic of $K$ admits finite characterizations?
    \item Which fragments of the modal language admit finite characterizations (w.r.t.~all frames)?
\end{enumerate}

We answer the first question, by showing that the modal logics that admit finite characterizations are precisely the \emph{locally tabular} modal logics. Since local tabularity is quite rare, we can view this as a negative result, which provides additional motivation for studying the second question.

In order to study the second question, we consider all modal fragments $\mathcal{L}_C$ that are generated by a
set of operators 
$C\subseteq\{\land,\lor,\Diamond,\Box,\top,\bot,\neg_{at}\}$,
where $\neg_{at}$ denotes atomic negation (i.e., negation applied to propositional variables).
We almost completely classify these fragments as to whether they admit finite characterizations.
In particular, our main result states that
$\mono$ (i.e., the positive fragment without $\top$ and $\bot$) is a maximal fragment admitting finite characterizations, and that finite
characterizations can be computed effectively for formulas in this fragment.

A significant part of the paper is devoted to the study of the fragment $\mono$. Building on~\cite{KurtoninaDeRijke}, we give a semantic characterization of this fragment in terms of 
invariance for ``weak simulations''. We then proceed to study the category of finite pointed Kripke structures under weak simulations. In 
particular, we establish a method for constructing
finite dualities (a.k.a.,~splittings) in this category. Finally, we show that these finite dualities give rise to finite characterisations for $\mono$-formulas. We also study size bounds for finite characterizations, and the computational complexity of computing these.


\paragraph{Outline}
In Section~\ref{sec:Preliminaries}, 
we review basic concepts and notation
from modal logic. 
In Section~\ref{sec:FinChars}, we introduce the notion of finite characterizations, and we present our main results regarding finite characterizations for modal formulas. Section~\ref{sec:mono}
is dedicated to the proof of our result for $\mono$, and also 
develops the general theory of this specific fragment. Finally,
in Section~\ref{sec:further}, we discuss
complexity aspects, and size bounds, and we discuss further extensions of our framework.

\paragraph{Related Work}

Most of the results presented here were obtained in the MSc thesis~\cite{MyThesis}. We draw on earlier results in~\cite{AlexeCKT2011,tencatedalmauCQ} regarding finite characterizations for database schema mappings and conjunctive queries (equivalently, positive-existential FO formulas). Specifically, as we will explain
 in Section~\ref{sec:FinChars}, it follows from results in~\cite{AlexeCKT2011,tencatedalmauCQ} that $\posex$ admits finite characterizations. Recently, in~\cite{FunkJL2022,cate2023rightadjoints}  the results from~\cite{AlexeCKT2011,tencatedalmauCQ} were further 
extended to
obtain finite characterizations in the presence of certain background theories such as lightweight ontologies. In~\cite{LTLWolter,PatrikThesis}, the authors studied the existence of finite characterisations for fragments of the linear temporal logic LTL. Finally,  \cite{StaworkoW15} studied finite characterizations for XPath twig-queries, which can also be viewed as a fragment of modal logic (with child and descendant modalities) interpreted on finite trees.

\section{Preliminaries (Modal Logic)}\label{sec:Preliminaries}
In this preliminary section, we will review the basics of modal logic. We refer to~\cite{bluebook} for more details.

\begin{definition}{(Full Modal Language, Modal Fragments)}
The \textit{full modal language} $\mathcal{L}_{\text{full}}[\Prop]$ over a set $\mathrm{Prop}$ of propositional variables is recursively generated by
\[\varphi::=\;\top\mid \bot\mid p\mid \neg_{\text{at}}p\mid \revision{(\varphi\wedge\varphi)}\mid \revision{(\varphi\vee\varphi)}\mid \Diamond\varphi\mid \Box\varphi\]
where $p\in\mathrm{Prop}$ and $\neg_{at}$ has the same semantics as ordinary negation $\neg$. A \textit{modal language} $\mathcal{L}$ over the variables \Prop~(notation: $\mathcal{L}[\Prop]$) is any collection of modal formulas over \Prop, i.e.~any subset of $\mathcal{L}_{\text{full}}[\Prop]$. Occasionally we may write just $\mathcal{L}_{\text{full}}$, when the set of base variables when it is clear from context.

For $S\subseteq\{\wedge,\vee,\Diamond,\Box,\top,\bot,\neg_{at}\}$, we denote by $\mathcal{L}_S[\Prop]$ fragment of
$\mathcal{L}_{\text{full}}[\Prop]$ consisting of formulas 
generated from propositional variables using only the operators in $S$. 
\end{definition}

We will refer to $\Box$ and $\Diamond$ as \emph{modal operators}. The \emph{dual} of a connective in 
$\{\wedge,\vee,\Diamond,\Box,\top,\bot\}$
is defined
as follows: $\land, \lor$ are each others dual, $\Box,\Diamond$ are each others dual, and $\top, \bot$ are each others dual.

Note that we assume formulas to be in \textit{negation normal form} (i.e. we assume that negations can only appear in front of variables, and that double negations never occur) and make this explicit by viewing $\neg_{at}$ as a restricted connective that can only be applied to (unnegated) variables. 




\begin{definition}{(Kripke Models, Pointed Models, Frames)}
A (Kripke) \textit{model} is a triple $M=(dom(M),R,v)$ where $dom(M)$ is a set of `possible worlds', $R\subseteq dom(M)\times dom(M)$ a binary `accessibility' relation and $v:dom(M)\to\mathcal{P}(\mathrm{Prop})$ is a colouring.\footnote{An ordinary valuation function $V:\mathrm{Prop}\to\mathcal{P}(dom(M))$ induces a `colouring' $col(V):dom(M)\to\mathcal{P}(\mathrm{Prop})$ as its \textit{transpose} $col(V)(w):=\{p\in\mathrm{Prop}\mid w\in V(p)\}$.} A (Kripke) \textit{frame} is a Kripke model without a colouring. For any model $M=(dom(M),R,v)$, let $\text{Fr}(M)=(dom(M),R)$ be its underlying frame. A \textit{pointed model} is a pair $(M,s)$ of a Kripke model $M$ together with a state $s\in dom(M)$.\footnote{Sometimes we will be sloppy and write $M,s$ instead of $(M,s)$ for pointed models.}
\end{definition}

\begin{definition}{(Semantics)}
The semantic clauses for modal logic are as follows:
\begin{align*}
    M,s\models p\qquad\textrm{iff}\qquad& p\in v(s)\\
    M,s\models\neg_{at}p\qquad\textrm{iff}\qquad&\textrm{not}\;M,s\models p\\
    M,s\models\top\qquad\;\;\;\qquad&\textrm{always}\\
    M,s\models\bot\qquad\;\;\;\qquad&\textrm{never}\\
    M,s\models\varphi\wedge\psi\qquad\textrm{iff}\qquad&M,s\models\varphi\;\textrm{and}\;M,s\models\psi\\
    M,s\models\varphi\vee\psi\qquad\textrm{iff}\qquad&M,s\models\varphi\;\textrm{or}\;M,s\models\psi\\
    M,s\models\Diamond\varphi\qquad\textrm{iff}\qquad&\exists t\in R[s]\;\textrm{with}\;M,t\models\varphi\\
    M,s\models\Box\varphi\qquad\textrm{iff}\qquad&\forall t\in R[s]\;\textrm{it holds that}\;M,t\models\varphi
\end{align*}

Finally, we set: 
\[\text{mod}(\varphi):=\{(M,s)\mid M,s\models\varphi\},\qquad \text{finmod}(\varphi):=\{(M,s)\mid M,s\models\varphi,\;dom(M)\;\text{is finite}\} \] 
\revision{For any two modal formulas $\varphi$ and $\psi$, we say that $\varphi$ entails $\psi$ and write $\varphi\models\psi$ if $mod(\varphi)\subseteq mod(\psi)$, i.e. if every model of $\varphi$ is also a model of $\psi$.} Finally, for any class of frames $\mathcal{F}$ we define
\[\text{mod}_{\mathcal{F}}(\varphi):=\{(M,s)\mid M,s\models\varphi\;\&\;\text{Fr}(M)\in\mathcal{F}\}\]
\end{definition}

\begin{definition}{(Bisimulation)}
Let $(M,s),(M',s')$ be two pointed models. A relation $Z\subseteq dom(M)\times dom(M')$ is a \textit{bisimulation} (notation: $Z: M,s\bisim M',s'$), \revision{and we say that $(M,s)$ and $(M',s')$ are bisimilar}, if $(s,s')\in Z$ and for all $(t,t')\in Z$, the following clauses are satisfied:
\begin{description}
    \item[(atom)] $v^M(t)=v^{M'}(t')$
    \item[(forth)] \revision{For all $u\in dom(M)$}, $R^{M}tu$ implies $\exists u'\in M'$ with $R^{M'}t'u'$ and $(u,u')\in Z$
    \item[(back)] \revision{For all $u'\in dom(M')$}, $R^{M'}t'u'$ implies $\exists u\in M$ with $R^{M}tu$ and $(u,u')\in Z$
\end{description}
\end{definition}

\begin{theorem}\label{thm:bisimulationinvariance}
Modal formulas are invariant under bisimulations, i.e.~if $M,s\models\varphi$ and $M,s\bisim M',s'$ then $M',s'\models\varphi$ (and vice versa) for all modal formulas $\varphi$. 
\end{theorem}

Kripke models can be seen as a special type of first order structures, and that the modal language can be seen as a \emph{fragment} of first order logic over these special types of structures through a fully compositional translation that is known as the \emph{standard translation} that maps a modal formula $\varphi$ to a first order formula $ST_x(\varphi)$ with one free variable, $x$. In fact, there is a well-known \textit{converse} to bisimulation-invariance, known as the \textit{Van Benthem Characterisation Theorem}.

\begin{theorem}[Van Benthem Characterisation \cite{benthem1976}]\label{thm:Johan}
A first-order formula $\alpha(x)$ (in the signature with one binary relation and finitely many unary relations) is invariant under bisimulation  iff it is equivalent to the standard translation of a modal formula.
\end{theorem}

For every pointed model $(M,s)$, let $(M\restriction s,s)$ be the \emph{point-generated submodel} of $M$
generated by $s$, i.e. the restriction of $M$ to the subset of all nodes reachable from $s$. It is easily seen that every pointed model $(M,s)$ is bisimilar to its point-generated submodel $(M\restriction s,s)$.
We say that a pointed model $(M,s)$ is point-generated if 
$M=M\restriction s$.
We say that a bisimulation $Z:M,s\bisim M',s'$
is \emph{total} if $dom(Z)=dom(M)$ and $ran(Z)=dom(M')$.

\begin{lemma}\label{lem:pointgeneratedbisimtotal}
If $Z:M,s\bisim M',s'$ and the pointed models $(M,s)$ and $(M',s')$ are point-generated, then $Z$ is total.
\end{lemma}
\begin{proof}
Let $Z:M,s\bisim M',s'$ where $M=M\restriction s$ and $M'=M'\restriction s'$. First, we show that for each $t\in dom(M)$ there is some $t'\in dom(M')$ with $(t,t')\in Z$. Since $(M,s)$ is point-generated, there is some finite path $(t_0,\ldots,t_n)$ in $M$ with $t_0=s$ and $t_n=t$. Applying the [forth] clause $n$ times we get a path $(t'_0,\ldots,t'_n)$ in $M'$ with $t'_0=s'$ and $(t_i,t'_i)\in Z$ for all $i\leq n$. Then $t'_n\in dom(M')$ is as desired. The proof in the other direction is entirely symmetric.
\end{proof}

Next, we define a number of notions that are
concerned with depth.

\begin{definition}{(Modal Depth)}
The \textit{modal depth} $d(\varphi)$ of a formula $\varphi$ is the length of the longest sequence of nested modal operators in $\varphi$.
\end{definition}

\begin{definition}\label{def:Paths}{(Path, Height)}
A \textit{finite path} through a model $M$ is a finite sequence $(t_0,\ldots,t_n)$ of states in $dom(M)$ such that $R(t_i,t_{i+1})$ holds for all $i<n-1$. We also write $t_0R\ldots Rt_n$ for this path. \revision{The length of a path is the number of states visited, i.e. the length of $(t_0,\ldots,t_n)$ is exactly $n-1$.} The height of a pointed model $(M,s)$ is the length of the longest finite path in $M$ starting at $s$, or $\infty$ if there is no finite upper bound.
\end{definition}

\revision{
\begin{definition}{($n$-Bisimulations)}
Let $(M,s),(M',s')$ be two pointed models and $n$ a natural number. We say that $(M,s)$ and $(M',s')$ are \emph{$n$-bisimilar}
(notation: $M,s\bisim_n\;M',s'$), if there is an indexed collection of
binary relations $Z_0,\ldots, Z_{n-1}\subseteq dom(M)\times dom(M')$  with $(s,s')\in Z_0$ such that
\begin{description}
    \item[(atom)] for all $(t,t')\in Z_0\cup\ldots\cup Z_{n-1}$ we have $v^M(t)=v^{M'}(t')$
\end{description}
and moreover for all $i< n-1$ and all $(t,t')\in Z_i$, the following clauses are satisfied:
\begin{description}
    \item[(forth)] For all $u\in dom(M)$, $R^{M}tu$ implies $\exists u'\in M'$ with $R^{M'}t'u'$ and $(u,u')\in Z_{i+1}$
    \item[(back)] For all $u'\in dom(M')$, $R^{M'}t'u'$ implies $\exists u\in M$ with $R^{M}tu$ and $(u,u')\in Z_{i+1}$
\end{description}
The collection $Z_0, \ldots, Z_{n-1}$ is called a $n$-\textit{bisimulation}.
\end{definition}
}

\begin{proposition}{\cite{bluebook}}\label{prop:boundeddepthpreservation}
Modal formulas of depth at most $n$ are preserved under $n$-bisimulations.
\end{proposition}

\begin{proposition}\label{prop:heightboundedbysize}
For all pointed models $(M,s)$, $height(M,s)\leq|dom(M)|$ or $height(M,s)=\infty$.
\end{proposition}
\begin{proof}
Suppose that $height(M,s)\ne\infty$, i.e.~$height(M,s)=n$ for some $n$. Then there is some path $(s,t_1\ldots,t_{n-1})$ of length $n$ through $M$ starting at $s$. All element on this
path must be distinct, because, otherwise, it is easy to see that
there would be an infinite path.
It follows that $height(M,s)=n\leq|dom(M)|$.
\end{proof}

\revision{
\begin{proposition}\label{prop:treeunravelling}
Let $n$ be a natural number. Every pointed model $(M,s)$ with $height(M,s)\geq n$ is $n$-bisimilar to a pointed model $(M',s')$ with $height(M',s')=n$.
\end{proposition}
\begin{proof}
We define the \emph{depth-$n$ tree unravelling} of a pointed model $(M,s)$, a well-known construction in modal logic \cite{bluebook}. Given $(M,s)$ we construct $(M',s')$ as follows. Set $dom(M')$ to be the set of all paths $(s_0,\ldots,s_k)$ through $M$ with $s_0=s$ and $k\leq n-1$, and let $s'\in dom(M')$ denote the singleton path $(s)$. The accessibility relation is given by initial segment, i.e. if $\rho=(s_0,\ldots,s_k)$ and $\rho'=(s'_0,\ldots,s'_{k'})$ are two such paths, then $R^{M'}\rho\rho'$ holds if $k'=k+1$, $s_i=s_i$ for all $i\leq k$ and $R^Ms_k,s_{k+1}$ holds. Finally, a proposition letter $p\in\Prop$ is true at a path $\rho=(s_0,s\ldots,s_k)$ if $p\in v^M(s_k)$. Clearly we have $height(M',s')=n$, since by assumption that $height(M,s)\geq n$ there exist at least one path of length $n$. For $i\leq n-1$ define the following relations:
\[Z_i:=\{(\rho,s)\in dom(M')\times dom(M) \mid s\;\text{is the last state occurring in} \;\rho\;\text{and}\;\rho\; \text{has length} \;i\}\]
We claim that $Z_0,\ldots Z_{n-1}$ forms an $n$-bisimulation between $(M',s')$ and $(M,s)$. Let $i\leq n-1$ and $(\rho,t)\in Z_i$. The atomic clause is in each case obvious by definition of the truth of atomic propositions in $M'$. For the [forth] clause, just note that if $R^{M'}\rho\rho'$ then by construction the last state of $\rho'$ is a successor of the last state of $\rho$. Lastly, for the [back] clause, if $R^Mtt'$ and $\rho=(s_0,\ldots,s_i)$ then there exists a successor $\rho':=(s_0,\ldots,s_i,s')$ of $\rho$ in $M'$ such that $R^{M'}\rho\rho'$.
\end{proof}
}

\section{Finite Characterizations for Modal Formulas}
\label{sec:FinChars}

In this section, we introduce the question of 
finite characterizability for modal formulas,
and we list the main results. Finite characterizations can be defined relative to any concept class. We therefore first briefly review the definition of a concept class, and we will explain how modal fragments can be viewed as concept classes.

\begin{definition}{(Concept Classes)}
A \emph{concept class} is a triple $\mathbb{C}=(X,C,\lambda)$, where $X$ is a class of \emph{examples}, $C$ is a collection of \revision{\emph{concepts}}\footnote{\revision{Although strictly speaking elements $c\in C$ are syntactic \emph{representations} of subsets $\lambda(c)$ of the example space $X$, we opt to call them concepts instead.}} and $\lambda:C\to 2^X$.
\end{definition}

We interpret $\lambda(c)$ as the collection of positive examples for $c$, and  the complement $X-\lambda(c)$ as the collection of negative examples for $c$. Note that we may have two distinct representations $c,c'\in C$ that represent the same concept, i.e., such that $\lambda(c)=\lambda(c')$. In this case, we call $c$ and $c'$ \emph{equivalent}. We will call a concept class $\mathbb{C}=(X,C,\lambda)$ \emph{essentially finite} if the set $\{\lambda(c)\;|\;c\in C\}$ is finite, or, in other words, 
if there are only finitely many concepts, up to equivalence.
 
\begin{definition}{(Finite characterisations)}
Let $\mathbb{C}=(X,C,\lambda)$ be any concept class.
\begin{itemize}
\item A \emph{collection of labeled examples} is a pair
 $(E^+,E^-)$, where $E^+,E^-\subseteq X$ are finite sets of examples.
 A concept $c\in C$ \textit{fits} $(E^+,E^-)$ if $E^+\subseteq\lambda(c)$ and $E^-\subseteq X\setminus \lambda(c)$, or, in other 
 words, if each $e\in E^+$ is a positive example for $c$ and each $e\in E^-$ is a negative example for $c$.
 \footnote{\revision{We will sometimes be sloppy in our notation and write that a concept $c\in C$ ``fits $E^+$ \rrevision{as positive examples}'' or ``fits $E^-$ \rrevision{as negative examples}'', by which we will mean that $c$ fits $(E^+,\emptyset)$ or $(\emptyset,E^-)$, respectively. }}
 
\item A \textit{finite characterisation} of a concept $c\in C$ w.r.t.~$\mathbb{C}$ is a collection of labeled examples \revision{$(E^+,E^-)$} such that (i) $c$ fits \revision{$(E^+,E^-)$} and (ii) $c$ is the only concept from $\mathbb{C}$, up to equivalence, that fits \revision{$(E^+,E^-)$} (i.e., 
for all $c'$ that fit \revision{$(E^+,E^-)$} we have $\lambda(c')=\lambda(c)$). 
\item We say that the concept class $\mathbb{C}$ \textit{admits finite characterisations} if each concept $c\in C$ has a finite characterisation w.r.t.~$\mathbb{C}$.
\end{itemize}
\end{definition}

Equivalently, a finite characterisation of a concept $c\in C$ is a pair of finite sets of examples \revision{$(E^+,E^-)$} such that $c$ fits $(E^+,E^-)$ and $(E^+,E^-)$ \emph{separates} $c$ from every non-equivalent concept in $C$ in the sense that for every $c'\in C$ with $\lambda(c)\ne\lambda(c')$ there is a \revision{positive example $x\in E^+$ with $x\not\in\lambda(c')$ or a negative example $x\in E^-$ such that $x\in\lambda(c')$.}

As an example, imagine the concept class built of closed intervals on the set integers $\mathbb{Z}$. \revision{That is, let $\mathbb{C}=(\mathbb{Z},\{[n,m]\;|\;n,m\in\mathbb{Z}\}, \lambda)$ where $\lambda([n,m])=\{z\in\mathbb{Z} \mid n\leq z\leq m\}$. Then, for instance, the interval $[4,6]$ has a finite characterisation w.r.t. $\mathbb{C}$ consisting of the positive examples $4$ and $6$ together with the negative examples $3$ and $7$.}


We will look at concepts corresponding to finite model classes $\text{finmod}(\varphi)$ of modal formulas $\varphi$, which defines a concept over the example space of all finite pointed models, possibly restricted to some frame class.

\begin{definition}
\label{def:modclass}
Let $\mathcal{F}$ be a class of frames, 
let $\Prop$ be a finite set of propositional variables, 
and let $\mathcal{L}[\Prop]\subseteq\mathcal{L}_{\text{full}}[\Prop]$ be a modal fragment.
We define the concept class
\[\mathbb{C}(\mathcal{F},\mathcal{L}[\Prop])=(X,C,\lambda)\]
where $X$ is the set of all pointed models over $\Prop$ whose frame belongs to $\mathcal{F}$, $C=\mathcal{L}[\Prop]$, and $\lambda(\varphi)=\{(M,s)\in X\mid M,s\models\varphi\}$.
%
\end{definition}

Let $\mathcal{F}_{\text{all}}$ denote the class of all frames and let $\mathcal{F}_{\text{fin}}$ denote the class of all finite frames. As noted in the introduction, from the point of view of practical applications, it is natural to require finite characterizations to consist of finite examples. This amounts
to restricting attention to finite frames. However, we will be 
chiefly concerned with modal logics that have the finite model 
property, and in this case, it is easy to see that every finite characterization over
$\mathcal{F}_{\text{fin}}$ is also a finite characterization
over $\mathcal{F}_{\text{all}}$. It is convenient for 
expository reasons to not enforce a restriction to finite 
frames in Definition~\ref{def:modclass}.

\subsection{The Full Modal Language}
\label{chap:FullLanguage}

We first show that the modal concept class $\mathbb{C}(\mathcal{F}_{\text{full}},\mathcal{L}_{\text{all}})$, which corresponds to the basic modal logic $\textbf{K}$, does not admit finite characterisations. 
We do this by showing that there are modal formulas that can force characterisations to contain examples of arbitrary height, thereby showing that these characterizations cannot be finite. 

\begin{lemma}\label{lem:height}
Let \revision{$height_n:=\Box^{n+1}\bot\wedge\Diamond^n\top$} (where $n\geq 0$). Then, for all pointed models $(M,s)$, we have that $M,s\models height_n$ iff $height(M,s)=n$.
\end{lemma}
\begin{proof}
\revision{Let $(M,s)$ be a pointed model. We have that $height(M,s)=n$ iff there is at least one path of length $n$, and no strictly longer path, emanating from $s$ through $M$. In particular, there cannot be any infinite paths and hence no loops in the model $M$. But clearly this hold iff $M,s\models\Diamond^n\top$ and $M,s\models\Box^{n+1}\bot$.}
\end{proof}


\begin{theorem}\label{thm:nocharswrtK}
Let $\Prop$ be any finite set. 
Then $\mathbb{C}(\mathcal{F}_{\text{fin}},\mathcal{L}_{\text{full}}[\Prop])$ does not admit finite characterizations. In fact, no modal formula has a finite characterisation w.r.t. $\mathbb{C}(\mathcal{F}_{\text{fin}},\mathcal{L}_{\text{full}}[\Prop])$. 
The same holds for $\mathbb{C}(\mathcal{F}_{\text{all}},\mathcal{L}_{\text{full}}[\Prop])$.
\end{theorem}
\begin{proof}
Let $\varphi\in\mathcal{L}_{\text{full}}[\Prop]$ and suppose that $\varphi$ has a finite characterisation $(E^+,E^-)$ w.r.t. $\mathbb{C}(\mathcal{F}_{\text{\text{fin}}},\mathcal{L}_{\text{full}}[\Prop])$. Observe that either (i) $\varphi\models\Box^k\bot$ for some $k$ or (ii) $\varphi\wedge\Diamond^k\top$ is satisfiable for each $k$. 

In case (i), by contraposition we get that $\Diamond^k\top\models\neg\varphi$. Let $n>\mathrm{max}(\{|dom(E)|\mid (E,e)\in E^-\}\cup\{k\})$ and \revision{define the formula $\varphi':=\varphi\vee height_n$. We claim that $\varphi'$ fits $(E^+,E^-)$ while $\varphi\not\equiv\varphi'$. It is obvious that $\varphi'$ fits $E^+$ \rrevision{as positive examples} because $\varphi\models\varphi'$. By definition, $\varphi$ is false on every example in $E^-$. Moreover, by choice of $n$ and Lemma \ref{lem:height}, $height_n$ is false on every example in $E^-$. It follows that $\varphi'$ fits $(E^+,E^-)$.} 
Moreover, \rrevision{let $(P_n,s)$ be the pointed model consisting of
a directed path of length $n$, where $s$ is the 
initial state. Note that $(P_n,s)$ has height $n$ and therefore satisfies $height_n$ and hence $\varphi'$. But $k\leq n$ which means that $height_n\models\Diamond^k\top$. Using our assumption that (i) $\varphi\models\Box^k\bot$ and contraposition we get that $P_n,s\models\neg\varphi$. Thus $P_n,s\models\varphi'\wedge\neg\varphi$ witnessing that $\varphi\not\equiv\varphi'$.}

In case (ii), let $n>\mathrm{max}(\{|dom(E)|\mid (E,e)\in E^+\}\cup\{d(\varphi)\})$, and consider the formula \revision{$\varphi':=\varphi\wedge\neg height_{n}$. We claim that $\varphi'$ fits $(E^+,E^-)$ but is non-equivalent to $\varphi$.} Clearly $\varphi'$ fits $E^-$ \rrevision{as negative examples}, \revision{since $\varphi'\models\varphi$. Further, by choice of $n$ and Lemma \ref{lem:height}, all positive examples in $E^+$ satisfy $\neg height_{n}$ (otherwise one of them would be of height $n$). Hence it follows that $\varphi'$ fits $(E^+,E^-)$.} But by hypothesis there is some pointed model $(M,s)$ satisfying $\varphi\wedge\Diamond^n\top$. \revision{By Proposition \ref{prop:treeunravelling}, there is a pointed model $(M',s')$ such that $M',s'\bisim_n\;M,s$ and $height(M',s')=n$. It follows from Proposition \ref{prop:boundeddepthpreservation} that $M',s'\models\varphi$. Moreover, by Lemma \ref{lem:height} we have $M',s'\models height_n$ and as $\varphi'\models\neg height_n$ we get that $M',s'\models\neg\varphi'$. Hence $M',s'\models\varphi\wedge\neg\varphi'$ witnessing that $\varphi\not\equiv\varphi'$.}

The corresponding result for the concept class $\mathbb{C}(\mathcal{F}_{\text{\text{all}}},\mathcal{L}_{\text{full}}[\Prop])$ follows by the finite model property of the full modal logic over the class of all frames, because this implies that any two modal formulas that can be separated on an infinite model can already be separated on a finite model.
\end{proof}


We go on to investigate which concept classes of the form $\mathbb{C}(\mathcal{F})$ admit finite characterisations.

\begin{remark}
The proof of Theorem~\ref{thm:nocharswrtK} also works, with minor changes, when restricting to transitive frames (corresponding to the modal logic $\textbf{K4}$). \revision{It suffices to replace $P_n$ by its transitive closure and observe that the height formulas preserve all their relevant properties over this frame class.} It follows that no modal formula has a finite characterisation w.r.t. $\mathbb{C}(\mathcal{F}_{\text{trans}})$ where $\mathcal{F}_{\text{trans}}$ is the class of all transitive frames. However, the proof does not extend to the class of reflexive and transitive frames (corresponding to the modal logic $\textbf{S4}$), or any other class of reflexive frames. \revision{This is because the formula $height_n$ is not satisfiable on reflexive frames.}
\end{remark}

Note that Theorem~\ref{thm:nocharswrtK} establishes a strong failure
of unique characterizability: not only does
$\mathbb{C}(\mathcal{F}_{\text{fin}},\mathcal{L}_{\text{full}}[\Prop])$ not admit unique characterization, in fact, \emph{no} formula is uniquely characterizable in this concept class. 

\begin{theorem}\label{thm:loctab_iff_finchar}
Let $\mathcal{F}$ be a class of frames and let $\Prop$ be a finite set.
Then the concept class $\mathbb{C}(\mathcal{F},\mathcal{L}_{\text{full}}[\Prop])$ admits finite characterisations iff it is essentially finite.
\end{theorem}


\begin{proof}
The right to left is easy: \revision{if $\{\text{mod}_{\mathcal{F}}(\varphi)\mid \varphi\in\mathcal{L}_{\text{full}}\}=\{c_1,\ldots,c_n\}$ is finite then for every $1\leq i,j\leq n$ with $i\ne j$ the symmetric difference $c_i\oplus c_j := (c_i-c_j)\cup(c_j-c_i)$ is non-empty.} For each such pair $(i,j)$, pick some example $e_{i,j}\in c_i\oplus c_j$. Then for each $c_i$, the set of examples $\{e_{i,j}\mid j\ne i\}$ uniquely characterises $c_i$ w.r.t. $\mathbb{C}(\mathcal{F})$.

For the left to right direction, suppose that $\mathbb{C}(\mathcal{F},\mathcal{L}_{\text{full}}[\Prop])$ admits finite characterisations. Then, in particular, the formula $\bot$ has a finite characterisation $(E^+,E^-)$, based on $\mathcal{F}$-frames.
Necessarily, $E^+=\emptyset$. It follows, by definition of finite characterisations, that, for every modal formula $\varphi$ that is satisfiable w.r.t.~$\mathcal{F}$, there must be some example $(E,e)\in E^-$ such that $E,e\models\varphi$. We claim that there can only be $2^{|E^-|}$ formulas up to logical equivalence over $\mathcal{F}$.
To see this, note that if there were more, then there would be two modal formulas  $\varphi,\psi$ with $\text{mod}_{\mathcal{F}}(\varphi)\ne \text{mod}_{\mathcal{F}}(\psi)$ such that $E^-\cap \text{finmod}_{\mathcal{F}}(\varphi)=E^-\cap \text{finmod}_{\mathcal{F}}(\psi)$. But then the formula $(\neg\varphi\wedge\psi)\vee(\varphi\wedge\neg\psi)$  is satisfiable, but not on any pointed model in $E^-$, a contradiction.
\end{proof}

This result can be equivalently stated in terms of 
\emph{local tabularity}. 
For a class of frames $\mathcal{F}$, we say that
``the modal logic of $\mathcal{F}$ is locally tabular'' if, for every finite $\Prop$, there are only finitely many modal formulas in $\Prop$ up to equivalence over $\mathcal{F}$~\cite{bluebook}.
Then we have:

\begin{theorem}
Let $\mathcal{F}$ be any class of frames.
Then the following are equivalent:
\begin{enumerate}
    \item For every finite set $\Prop$, $\mathbb{C}(\mathcal{F},\mathcal{L}_{\text{full}}[\Prop])$ admits finite characterisations,
    \item The modal logic of $\mathcal{F}$ is locally tabular.
\end{enumerate}
\end{theorem}

In summary, no restriction to a smaller class of frames can make the full modal language admit finite characterisations, except by trivializing the problem.

\subsection{Fragments of Modal Logic}\label{chap:fragments}

In this section, we investigate which modal concept classes of the form $\mathbb{C}(\mathcal{F},\mathcal{L}_{S}[\Prop])$ admit finite characterisations, where $S\subseteq\{\Box,\Diamond,\wedge,\vee,\top,\bot,\neg_{at}\}$ and where
$\mathcal{F}$ is either $\mathcal{F}_{\text{fin}}$ or $\mathcal{F}_{\text{all}}$. In fact, 
since the modal logic of $\mathcal{F}_{\text{fin}}$ coincides
with the modal logic of $\mathcal{F}_{\text{all}}$~(cf.~\cite{bluebook}), we have that, whenever a set of positive and negative examples 
uniquely characterizes a modal formula w.r.t.~$\mathbb{C}(\mathcal{F}_{\text{fin}},\mathcal{L}_{S}[\Prop])$, it
also uniquely characterizes the same formula w.r.t.~$\mathbb{C}(\mathcal{F}_{\text{all}},\mathcal{L}_{S}[\Prop])$. 
In particular, this means that,
whenever $\mathbb{C}(\mathcal{F}_{\text{fin}},\mathcal{L}_{S}[\Prop])$ admits finite characterisations, so does
 $\mathbb{C}(\mathcal{F}_{\text{all}},\mathcal{L}_{S}[\Prop])$. We will simply write
$\mathbb{C}(\mathcal{L}_{S}[\Prop])$ in the remainder of this 
section, with the understanding that all our positive results will
be proved for $\mathcal{F}_{\text{fin}}$ and all our negative 
results will be proved for $\mathcal{F}_{\text{all}}$.

For our first positive result, we make use of known results regarding
finite characterizations for positive-existential first-order
formulas. Specifically, 
it was shown in~\cite{AlexeCKT2011,tencatedalmauCQ} that 
the concept class of \emph{c-acyclic unions of conjunctive queries}, which can be seen as a restricted class of positive-existential first-order formulas, admits finite characterisations, and that finite
characterizations for such formulas can be constructed effectively.
As it turns out, the standard translation of every 
$\posex$-formula can be equivalently written as a c-acyclic union of conjunctive queries (cf.~\cite{MyThesis} for full details). From this, we obtain:


\begin{theorem}\label{thm:Balder}
For every finite set $\Prop$, $\mathbb{C}(\mathcal{L}_{\Diamond,\wedge,\vee,\top,\bot}[\Prop])$ admits finite characterisations. Moreover, these characterisations are
effectively computable. 
\end{theorem}
\begin{proof}
It follows immediately from the aforementioned results in~\cite{AlexeCKT2011,tencatedalmauCQ} that 
$\mathbb{C}(\posex[\Prop])$ admits unique characterizations. 
It is easy to see that every $\varphi\in\mathcal{L}_{\Diamond,\wedge,\vee,\top,\bot}[\Prop]$ is equivalent to a formula in $\posex[\Prop]\cup\{\bot\}$. 
Therefore, 
to extend this
result to $\mathcal{L}_{\Diamond,\wedge,\vee,\top,\bot}$,
it suffices to
\begin{enumerate}
\item[(i)] show that every
finite characterization $(E^+,E^-)$ for a formula 
$\varphi\in \posex[\Prop]$ w.r.t. the concept class $\mathbb{C}(\posex[\Prop])$ can be effectively extended to finite characterisation of $\varphi$ w.r.t.~$\posex[\Prop]\cup\{\bot\}$, and
\item[(ii)] to provide a finite characterisation of $\bot$ w.r.t.~$\mathbb{C}(\mathcal{L}_{\Diamond,\wedge,\vee,\top,\bot}[\Prop])$.
\end{enumerate}
\revision{For the first item, it is enough to ensure that $E^+$ is non-empty.
This can be done, for instance,}
by adding to $E^+$ the pointed model
$\fullloop$, i.e.,  single-element reflexive pointed model
satisfying all propositional variables. Note that this pointed
model satisfies all
$\posex$-formulas.
For the second, it suffices to use $\fullloop$ as a negative
example.
\end{proof}

Next, we show that the dual fragment $\mathcal{L}_{\Box,\wedge,\vee,\top,\bot}$ also admits finite characterisations. 

\begin{theorem}\label{thm:dualposex}
For every finite set $\Prop$, $\mathbb{C}(\mathcal{L}_{\Box,\wedge,\vee,\top,\bot}[\Prop])$ admits finite characterisations. Moreover, these characterisations are effectively computable.
\end{theorem}
\begin{proof}
Consider any formula $\varphi\in\mathcal{L}_{\Box,\wedge,\vee,\top,\bot}[\Prop]$. Define the formula $\triangleleft(\varphi)$ to be the formula obtained from $\varphi$ by replacing each connective by its dual \rrevision{(cf.~also page~\pageref{dualoperation})}. It follows that $\triangleleft(\varphi)\in\mathcal{L}_{\Diamond,\wedge,\vee,\top,\bot}[\Prop]$, so by theorem \ref{thm:Balder} $\triangleleft(\varphi)$ has a finite characterisation $(E^+_{\triangleleft(\varphi)}, E^-_{\triangleleft(\varphi)})$ w.r.t. $\mathcal{L}_{\Diamond,\wedge,\vee,\top,\bot}[\Prop]$. Now define
\[E^+_{\varphi}:=\{(M^{\neg},s)\mid (M,s)\in E^-_{\triangleleft(\varphi)}\}\qquad E^-_{\varphi}:=\{(M^{\neg},s)\mid (M,s)\in E^+_{\triangleleft(\varphi)}\}\]
where $M^{\neg}=(W,R,v')$ is obtained from $M=(W,R,v)$ by flipping the colouring of each state $t$ in $M$, i.e. setting $v'(t)=\Prop-v(t)$. Furthermore, for every modal formula $\varphi$, if we let $\varphi^{\neg}$ be the formula obtained from $\varphi$ by substituting every positive occurrence of a propositional variable with its negation then we have that for every model $(M,s)$;
\[M,s\models\varphi\qquad\text{iff}\qquad M^\neg,s\models\varphi^{\neg}\]

We claim that $(E^+_\varphi, E^-_\varphi)$ is a finite
characterization for $\varphi$.
Suppose that some $\psi\in\mathcal{L}_{\Box,\wedge,\vee,\top,\bot}[\Prop]$ fits $(E^+_{\varphi},E^-_{\varphi})$. This means that $\psi$ is true on all $(M^\neg,s)$ for $(M,s)\in E^-_{\triangleleft(\varphi)}$ and $\psi$ is false on all $(M^\neg,s)$ for $(M,s)\in E^+_{\triangleleft(\varphi)}$. It follows that $\neg\psi$ is false on all $(M^\neg,s)$ for $(M,s)\in E^-_{\triangleleft(\varphi)}$ and $\neg\psi$ is true on all $(M^\neg,s)$ for $(M,s)\in E^+_{\triangleleft(\varphi)}$.
Thus, by the above equivalence we get that \rrevision{$\neg\psi^{\neg}\in\posex[\Prop]$} fits $(E^+_{\triangleleft(\varphi)}, E^-_{\triangleleft(\varphi)})$. But observe that $\neg\psi^{\neg}\equiv\triangleleft(\psi)$, so in fact $\triangleleft(\psi)$ fits $(E^+_{\triangleleft(\varphi)},E^-_{\triangleleft(\varphi)})$ and thus $\triangleleft(\varphi)\equiv\triangleleft(\psi)$ because $(E^+_{\triangleleft(\varphi)},E^-_{\triangleleft(\varphi)})$ is a finite characterisation. Finally, it rests to observe that $\chi\equiv\xi$ iff $\triangleleft(\chi)\equiv\triangleleft(\xi)$ holds for all modal formulas $\chi,\xi$, and we conclude that $\varphi\equiv\psi$.

We will make use of the above dualization construction  again later in Section 
\ref{sec:mono}, where we discuss it in detail (cf.~Definitions \ref{def:negonformulas} and \ref{def:negonmodels} and Proposition \ref{prop:coherencenegoperations}).
\end{proof}

This is one example of a general phenomenon, for every modal fragment of the form $\mathcal{L}_S$, $\mathbb{C}(\mathcal{L}_S)$ admits finite characterisations iff $\mathbb{C}(\mathcal{L}_{S^{\text{dual}}})$ admits finite characterisations, where $S^{\text{dual}}$ is obtained from $S$ by replacing connectives by their dual.


Next, we consider what happens for fragments containing both $\Diamond$ and $\Box$. 


\begin{restatable}{theorem}{thmmainmono}
\label{thm:main-mono}
For every finite set $\Prop$,
$\mathbb{C}(\mono[\Prop])$ admits finite characterisations. Moreover, these characterisations are effectively computable. 
\end{restatable}

The proof of Theorem~\ref{thm:main-mono} is quite different, and more involved. It will be given in the next section. We will discuss computational complexity aspects in Section \ref{sec:complexity}.

The following result shows that Theorem~\ref{thm:main-mono}
is optimal, in the sense that $\mono$ cannot be extended
to a larger fragment admitting finite characterisations.

\begin{theorem}\label{thm:negative-result-fragments}
Let $\Prop$ be any finite set. Then neither $\mathbb{C}(\mathcal{L}_{\Box,\Diamond,\wedge,\bot}[\Prop])$ nor $\mathbb{C}(\mathcal{L}_{\Box,\Diamond,\vee,\top}[\Prop])$ admits finite characterisations.
\end{theorem}
\begin{proof}
Observe that $height_n$ \revision{can be equivalently written as $\Box^{n+1}\bot\wedge\Diamond^n\Box\bot$, which belongs to $\mathcal{L}_{\Box,\Diamond,\wedge,\bot}[\emptyset]$, and that its negation $\neg height_{n}$ can be equivalently written as $\Diamond^{n+1}\top\vee\Box^n\Diamond\top$, which belongs to $\mathcal{L}_{\Box,\Diamond,\vee,\top}[\emptyset]$.} In the proof of Theorem~\ref{thm:nocharswrtK}, in order to show
that no modal formula admits a finite characterisations, we used both $height_n$ and its negation $\neg height_{n}$. However, here, we are only claiming that the concept classes in question do not admit finite characterisations, and therefore, we only have to give one example of a non-finitely-characterisable formula. 

First, we show that $\bot\in\mathcal{L}_{\Box,\Diamond,\wedge,\bot}[\emptyset]$ does not have a finite characterisation w.r.t. $\mathbb{C}(\mathcal{L}_{\Box,\Diamond,\wedge,\bot}[\emptyset])$. It follows that it also does not have a finite characterisation w.r.t. $\mathbb{C}(\mathcal{L}_{\Box,\Diamond,\wedge,\bot}[\Prop])$ for $\Prop\ne\emptyset$. Suppose for contradiction that $\bot$ did have such a characterisation $(E^+,E^-)$. Since $\bot$ has no positive examples, it must be that $E^+=\emptyset$ and $E^-$ is a finite set of examples with the property that every satisfiable $\mathcal{L}_{\Box,\Diamond,\wedge,\bot}[\emptyset]$ formula is satisfied on some example in $E^-$. \revision{Choose $k$ to be any number strictly greater than the size of any example in $E^-$. Then by Proposition \ref{prop:heightboundedbysize} and choice of $k$ we have that $height(E,e)<k$ for all $(E,e)\in E^-$. It follows that the formula $height_k$ is false on every example in $E^-$. But then $height_k\in\mathcal{L}_{\Box,\Diamond,\wedge,\bot}[\emptyset]$ fits $(E^+,E^-)$ whereas clearly $\bot\not\equiv height_k$.}


Similarly, we show that $\top\in\mathcal{L}_{\Box,\Diamond,\vee,\top}[\emptyset])$ cannot have a finite characterisation w.r.t. $\mathbb{C}(\mathcal{L}_{\Box,\Diamond,\vee,\top}[\emptyset])$, from which the more general result for arbitrary $\Prop$ follows. Note that, in principle, this already follows from the duality principle alluded to above. Suppose for contradiction that $\top$ has such a finite characterisation $(E^+,E^-)$. Since $\top$ is unfalsifiable, it follows that $E^-=\emptyset$ and $E^+$ is a finite set of examples with the property that every falsifiable formula in $\mathcal{L}_{\Box,\Diamond,\vee,\top}[\emptyset])$ is falsified on some example in $E^+$. \revision{Choose $k$ to be any number strictly greater than the size of any example in $E^+$. Then by Proposition \ref{prop:heightboundedbysize} and choice of $k$ we have that $height(E,e)<k$ for all $(E,e)\in E^+$. It follows that the formula $\neg height_{k}$ is true on all examples in $E^+$, because it can only be falsified on models of exactly height $k$. But then $\neg height_{k}\in\mathcal{L}_{\Box,\Diamond,\vee,\top}[\emptyset])$ fits $(E^+,E^-)$ whereas clearly $\top\not\equiv\neg height_{k}$.}
\end{proof}

Observe that ``admitting finite characterisations'' is a monotone property of modal fragments: if $\mathcal{L}\subseteq \mathcal{L}'$
and
$(\mathcal{F},\mathcal{L'}[\Prop])$ admits finite characterizations,
then so does $(\mathcal{F},\mathcal{L}[\Prop])$.
Therefore, Theorem~\ref{thm:negative-result-fragments} implies that
Theorem~\ref{thm:main-mono} cannot be improved (because $\mathcal{L}_{\Box,\Diamond,\vee,\top}\subseteq\mathcal{L}_{\Box,\Diamond,\wedge,\vee,\top}$ and $\mathcal{L}_{\Box,\Diamond,\wedge,\bot}\subseteq\mathcal{L}_{\Box,\Diamond,\wedge,\vee,\bot}$, while $\mathcal{L}_{\Box,\Diamond,\wedge,\vee,\neg_{\text{at}}}$ is
expressively equivalent to $\mathcal{L}_{\text{full}}$). The results above almost completely determine which concept classes of the form $\mathbb{C}(\mathcal{L}_S)$ admit finite characterisations (where $S\subseteq\{\Box,\Diamond,\wedge,\vee,\top,\bot,\neg_{\text{at}}\}$. \revision{In essence, the remaining open cases are $\mathcal{L}_{\Box,\Diamond,\wedge,\top},\mathcal{L}_{\Diamond,\wedge,\vee,\top,\bot,\neg_{\text{at}}}$ and $\mathcal{L}_{\Box,\Diamond,\top,\bot,\neg_{\text{at}}}$. We hope to settle these open problems in future work.}


\begin{remark}
Note that the fragments for which we proved above that they admit finite characterizations, are not locally tabular. The fragments without any modal operators are locally tabular, but once a fragment contains even a single modal operator, say $\Diamond$, there are infinitely many pairwise-inequivalent formulas of the form $\Diamond^np$.
\end{remark}



\section{Finite Characterisations for \texorpdfstring{$\mono$}{L(Box,Diamond,And,Lor)} via Weak-Simulation Dualities}\label{sec:mono}
In this section, we study the fragment $\mono$ in depth. We introduce 
\emph{weak simulations} and show that $\mono$ is characterized by its preservation under weak simulations. We study the category $\mathbf{wSim}[\Prop]$
of pointed models and weak simulations,
and we show how to construct finite dualities in 
this category. Finally, we use these results to prove Theorem~\ref{thm:main-mono}.


\subsection{Weak Simulations}

Kurtonina and de Rijke~\cite{KurtoninaDeRijke}
studied the positive modal fragment $\pos$
(which includes $\top$ and $\bot$). In order
to characterize the expressive power of this
fragment, they introduced the notion of a
\emph{directed simulation} (which we will simply call \emph{simulation} here, for brevity). We briefly recall notion
and the main result from~\cite{KurtoninaDeRijke}.

\begin{definition}{(Simulation)}
Given two pointed models $(M,s),(M',s')$, a \textit{simulation} from $(M,s)$ to $(M',s')$ is a relation $Z\subseteq M\times M'$ such that for all $(t,t')\in Z$, the following conditions hold.
\begin{description}
    \item[(atom)] $v^M(t)\subseteq v^{M'}(t')$
    \item[(forth)] If $R^{M}tu$, then there is a $u'$ with $R^{M'}t'u'$ and $(u,u')\in Z$
    \item[(back)] If $R^{M'}t'u'$, then there is a $u$ with $R^{M}tu$ and $(u,u')\in Z$
\end{description}
\end{definition}

The difference between simulations and bisimulations
lies in the ``atom'' clause: simulations only require
propositional variables to be preserved from left to right, while bisimulations require the propositional variables to 
be preserved in both directions. 

We say that a modal formula $\varphi$ is 
\emph{preserved under simulations} if,
whenever $Z$ is a simulation from $(M,s)$ to $(M',s')$ and $(M,s)\models\varphi$, then
$(M',s')\models\varphi$.

\begin{theorem}{(\cite{KurtoninaDeRijke})}
A first order formula $\varphi(x)$ (in the signature with one binary relation and finitely many unary relations) is preserved under simulations iff it is equivalent to the standard translation of an $\pos$ formula.
\end{theorem}

We can think of simulations as 
a weak form of bisimulations,
and we can view the above result as an 
analogue of the classic Van Benthem characterization \revision{(Theorem \ref{thm:Johan})}, which characterizes the full modal language in terms of \revision{invariance under bisimulations}. In order to characterize the expressive power
of $\mono$ (which differs from $\pos$ in that it lacks $\top$ and $\bot$) we will now introduce a modified version of simulations, which we call \emph{weak simulations}.

\begin{definition}{(Empty and Full Loopstates)}
The \emph{empty loopstate}, denoted by $\emptyloop$, is the pointed model consisting of a single reflexive point with an empty valuation. Similarly, the \emph{full loopstate}, denoted by $\fullloop$, is the pointed model consisting of a single reflexive point with a full valuation (i.e., where every $p\in\Prop$ is true).
\end{definition}

\begin{definition}\label{def:wSims}{(Weak Simulation)}
Given two pointed models $(M,s),(M',s')$, a \textit{weak simulation} between them (notation: $M,s\simu M',s'$, \revision{and we say that $(M',s')$ weakly simulates $(M,s)$ or $(M,s)$ is weakly simulated by $(M',s')$}) is a relation $Z\subseteq M\times M'$ such that for all $(t,t')\in Z$;
\begin{description}
    \item[(atom)] $v^M(t)\subseteq v^{M'}(t')$
    \item[(forth')] If $R^{M}tu$, either $M,u\bisim\emptyloop$ or there is a $u'$ with $R^{M'}t'u'$ and $(u,u')\in Z$
    \item[(back')] If $R^{M'}t'u'$, either $M',u'\bisim\fullloop$ or there is a $u$ with $R^{M}tu$ and $(u,u')\in Z$
\end{description}
\end{definition}

We refer to the extra `loopstate-conditions' as the \textit{escape clauses}, because, intuitively, they exclude $\emptyloop$ and $\fullloop$ from the forth and back requirements, respectively. The following two lemmas explicate some consequences of having these escape clauses. 

\begin{lemma}\label{lem:Winitialemptyloop}
Every pointed model weakly simulates $\emptyloop$ and is weakly simulated by $\fullloop$.
\end{lemma}
\begin{proof}
Let $(M,s)$ be a pointed model. We claim that $Z:=dom(\emptyloop)\times dom(M)$ is a weak simulation $Z:\emptyloop\simu M,s$. Let $(r,t)\in Z$ be arbitrary. [atom] trivial since $v(r)=\emptyset$. [forth'] trivial since we can always use the escape clause. [back'] Whenever $Rt'u'$ we can always choose $r$ itself as matching successor as $Rrr$ holds and $(r,u')\in Z$. A dual argument shows that $dom(M)\times dom(\fullloop)$ is the unique weak simulation $M,s\simu\fullloop$.
\end{proof}

In particular, the deadlock model (i.e.~the single point model with no successors) \textit{weakly} simulates $\emptyloop$ and is \textit{weakly} simulated by $\fullloop$, whereas it does not simulate $\emptyloop$ nor is simulated by $\fullloop$. 

\begin{lemma}\label{lem:emptyloopsatall}
The empty loopstate $\circlearrowright_{\emptyset}$ satisfies no formula in $\mono$, while the full loopstate $\circlearrowright_{\mathrm{Prop}}$ satisfies every formula in $\mono$.
\end{lemma}
\begin{proof}
Clearly $\emptyloop\not\models p$ and $\fullloop\models p$ for every atomic formula $p$. Now suppose that $\emptyloop\not\models\varphi$ and $\fullloop\models\varphi$. It follows that $\emptyloop\not\models\Diamond\varphi$ and that $\emptyloop\not\models\Box\varphi$. This is because there is exactly one successor that is bisimilar to $\emptyloop$ again, which by inductive hypothesis and Theorem \ref{thm:bisimulationinvariance} does not satisfy $\varphi$. Dually, $\fullloop\models\Diamond\varphi$ and $\fullloop\not\models\Box\varphi$ because there is exactly one successor which by inductive hypothesis and Theorem \ref{thm:bisimulationinvariance} satisfies $\varphi$.
\end{proof}

We say that a modal formula $\varphi$ is \emph{preserved under weak simulations}, if, whenever $M,s\simu M',s'$  and $M,s\models\varphi$, then $M',s'\models\varphi$.
Lemma~\ref{lem:emptyloopsatall} plays an important role in the next proof.

\begin{theorem}\label{thm:wSimPreservation}
Every formula $\varphi\in\mono$ is preserved under weak simulations.
\end{theorem}
\begin{proof}
Let $Z:M,s\simu M',s'$ be a weak simulation. We proceed by formula induction on $\varphi$. 
The base case is where $\varphi=p$ for some $p\in\mathrm{Prop}$. Then $M,s\models p$ implies $M',s'\models p$ as $v(s)\subseteq v(s')$.
The induction argument for conjunction and disjunction is straightforward. This leaves us with the modal operators.
\begin{itemize}
    \item[($\Diamond$)] Suppose that $M,s\models\Diamond\varphi$ where $\varphi\in\mono$, i.e.~there is some $t\in R^M[s]$ with $M,t\models\varphi$. By lemma ~\ref{lem:emptyloopsatall}, $M,t\negbisim\emptyloop$ and hence by [forth'] there must be some $t'\in R^{M'}[s']$ such that $(t,t')\in Z$ and hence $M',t'\models\varphi$ by inductive hypothesis. It follows by the semantics that $M',s'\models\Diamond\varphi$.
    \item[($\Box$)] Suppose that $M,s\models\Box\varphi$ with $\varphi\in\mono$. We show that $M',s'\models\Box\varphi$. If $R^{M'}[s']=\emptyset$ this holds vacuously. For an arbitrary $t'\in R^{M'}[s']$, by the [back'] clause either $M',t'\bisim\fullloop$ or there is a successor $t\in R^M[s]$ with $(t,t')\in Z$. In the former case, by Lemma \ref{lem:emptyloopsatall} $M',t'\models\varphi$. Else, as $M,s\models\Box\varphi$ we have $M,t\models\varphi$ and hence by inductive hypothesis $M',t'\models\varphi$.
\end{itemize}
\end{proof}

There exist modal formulas that are preserved under
weak simulations but that cannot be expressed in $\mono$.
Indeed, $\top$ and $\bot$ are preserved under weak simulations (actually, they are preserved by every relation) but it follows from Lemma~\ref{lem:emptyloopsatall} that they are not equivalent to any formula in $\mono$. However, we will show that the language $\mono\cup\{\top,\bot\}$ does satisfy a van Benthem-style characterisation  in terms of preservation under weak simulations. To that end, we will need the following lemma.
Recall that $\bisim$ denotes bisimulation.

\begin{lemma}\label{lem:emptytypemeansemptyloop}
For any pointed model $(M,s)$, 
\begin{itemize}
    \item $M,s\bisim\emptyloop$ iff, for all $\varphi\in\mono$, $M,s\not\models\varphi$ 
    \item $M,s\bisim\fullloop$ iff, for all $\varphi\in\mono$, $M,s\models\varphi$.
    \end{itemize}
\end{lemma}
\begin{proof}
The left to right direction is immediate from Lemma \ref{lem:emptyloopsatall} and Theorem \ref{thm:bisimulationinvariance}. 

From right to left, suppose that $M,s\not\models\varphi$ for all $\varphi\in\mono$. Now let $M'$ be the submodel of $M$  generated by $s$, and set $Z:=dom(M')\times dom(\emptyloop)$, where $dom(\emptyloop)=\{r\}$ consists of a single point. We show this is a bisimulation $Z:M',s\bisim\emptyloop$. Take any $(t,r)\in Z$. By construction, there is a path $s=t_0R\ldots Rt_n=t$ in $M'$ from $s$ to $t$ of length $n$ for some $n$. 

[atom] Note that $v(r)=\emptyset$ has empty colour so we want to show that $v(t)$ is empty as well. Suppose otherwise, i.e. there is some $p\in v(t)$, then it follows that $M,s\models\Diamond^np$, contradicting our assumption that $(M,s)$ refutes every formula in $\mono$. 

[forth] If $R^Mtt'$, then $t'\in dom(M')$ so $(t',r)\in Z$ by definition of $Z$ and $r$ sees itself. 

[back] If $Rrr'$ then $r'=r$ as $r$ only sees itself. We know that $M,t\not\models\Box p$ otherwise $M,s\models\Diamond^n\Box p$ while we know that $(M,s)$ refutes every formula in $\mono$. Hence $M,t\models\Diamond\neg p$ and thus has some successor $t'\in dom(M')$ for which $(t',r)$ automatically holds.

For the other equivalence, suppose that $M,s\models\varphi$ for all $\varphi\in\mono$. Again, let $M'$ be the  submodel of $M$ generated by $s$ in $M$, and set $Z:=dom(M')\times dom(\fullloop)$, where $dom(\fullloop)=\{r\}$ consists of a single point. We show this is a bisimulation $Z:M',s\bisim\fullloop$. Take any $(t,r)\in Z$. By construction, there is a path $s=t_0R\ldots Rt_n=t$ in $M'$ from $s$ to $t$ of length $n$ for some $n$. 

[atom] Note that $v(r)=\mathrm{Prop}$ has full colour so we want to show that $v(t)$ is full as well. Suppose otherwise, i.e. there is some $p\not\in v(t)$, then it follows that $M,s\not\models\Box^n p$, contradicting our assumption that $(M,s)$ satisfies every formula in $\mono$. 

[forth] If $R^Mtt'$, then $t'\in dom(M')$ so $(t',r)\in Z$ by definition and $r$ sees itself. 

[back] $r$ \textit{only} sees itself so if $Rrr'$ then $r'=r$. Now we know that $M,t\models\Diamond p$ otherwise $M,s\not\models\Box^n\Diamond p$ while we know that $(M,s)$ satisfies every formula in $\mono$. Thus $t$ must have some successor $t'\in dom(M')$ for which $(t',r)$ automatically holds.
\end{proof}

\revision{As is usual in semantic characterisations, we make an appeal to the existence of $\omega$-saturated elementary extension for every first order structure. In particular, for our proof we need \emph{modally saturated} Kripke models. A Kripke model $M$ is said to be modally saturated if for every $t\in dom(M)$ and infinite set of modal formulas $\Phi$, if every finite subset of $\Phi$ is satisfied at some successor of $t$, then there exists a successor of $t'$ of $t$ in $M$ such that $M^*,t\models\varphi$ for all $\varphi\in\Phi$.}

\begin{lemma}\label{lem:omegasaturated}
\revision{For every pointed Kripke model $(M,s)$, there exists a modally saturated pointed Kripke model $(M^*,s^*)$ such that $(M,s)$ and $(M^*,s^*)$ satisfy the same first order formulas with one free variable.}
\end{lemma}

\begin{proof}
\revision{
This is well known, and follows from classic results in model theory, 
since $\omega$-saturation implies modal saturation. For the sake of 
completeness, we spell out the details.

Let a pointed Kripke model $(M,s)$ be given. Identifying Kripke models with first order structure over the appropriate signature, we use a well-known result from model theory that says that every first order structure has an $\omega$-saturated elementary extension (see \cite{ChangKeisler}). It follows that we obtain an $\omega$-saturated first order structure, which corresponds to a Kripke model $M^*$, together with a map $f:dom(M)\to dom(M^*)$ such that $M\models\beta(t)$ iff $M^*\models\beta(f(t))$ for all $t\in dom(M)$ 
\rrevision{and for all
first-order formulas $\beta(x)$ in one free variable.}
We set $s^*:=f(s)$. 

To see that $M^*$ is modally saturated, let $\Phi$ be an infinite set of modal formulas and consider the infinite set of first order formulas $\Phi':=\{ST_x(\varphi)\:|\;\varphi\in\Phi\}\cup\{Rcx\}$ over the corresponding first order signature extended with a constant $c$ which is interpreted as the state $s^*$. Using the fact that the first order structure corresponding to $M^*$ is $\omega$-saturated, it follows that $\Phi'$ (all whose formulas contain a single free variable $x$) is satisfied at some state $u\in dom(M^*)$, so $M^*,u\models\varphi(x)$ for all $\varphi\in\Phi$. But as $c$ is interpreted as $s^*$, it follows that $u$ is a successor of $s^*$. This shows that $M^*$ is modally saturated.}
\end{proof}

\begin{theorem}
\label{thm:wsim-characterization}
A first order formula \revision{$\alpha(x)$} (in the signature with one binary relation and finitely many unary relations) is preserved under weak simulations iff it is equivalent to the standard translation of a formula in $\mono\cup\{\top,\bot\}$.
\end{theorem}

\begin{proof}
We already established the right-to-left direction (Theorem~\ref{thm:wSimPreservation}). For the left-to-right direction, let $\alpha(x)$ be some weak simulation-preserved first order formula, and set $\text{mon-con}(\alpha):=\{ST_x(\psi)\mid \psi\in\mono\cup\{\top,\bot\}\;\&\;\alpha\models ST_x(\psi)\}$ \revision{(where ``$\text{mon-con}$'' stands for monotone consequences)}.
The argument follows essentially the same line as in the characterisation for $\pos$ in \cite{KurtoninaDeRijke}. Clearly $\alpha$ implies every formula in $\text{mon-con}(\alpha)$ by definition. We will show that also $\text{mon-con}(\alpha)\models\alpha$. From this, it follows from a compactness argument that there are $\psi_1,\ldots,\psi_n\in\mono\cup\{\top,\bot\}$ such that $ST_x(\psi_1\wedge\ldots\wedge\psi_n)\equiv\alpha$. Note that in case $\alpha=\top$, $n=0$ but in that case $\psi_1\wedge\ldots\wedge\psi_n$ is the empty conjunction which is $\top$ and $ST_x(\top):=\top\in\mono\cup\{\top,\bot\}$. To show this entailment, let $M,s\models\text{mon-con}(\alpha)$,
\rrevision{in other words, let $M$ be any structure satisfying the set of first-order formulas $\text{mon-con}(\alpha)$ under the assignment that sends $x$ to $s$.}

\rrevision{
Define the set of first order formulas $\Phi:=\{\neg ST_x(\psi)\mid \psi\in\mono \text{ and } M,s\models\neg\psi\}\cup\{\alpha\}$. By a standard model-theoretic argument, $\Phi$ is satisfiable. For suppose otherwise, then by compactness there are $\varphi_1,\ldots,\varphi_n\in\mono$ such that $M,s\not\models\varphi_i$ (for all $i=1\ldots n$) and such that $\{\alpha,\neg ST_x(\varphi_1),\ldots$ $,\neg ST_x(\varphi_n)\}$ is inconsistent.} Since $ST_x(\cdot)$ commutes with the Boolean connectives, it follows that $\alpha\models ST_x(\varphi_1\vee\ldots\vee\varphi_n)$, where thus $(\varphi_1\vee\ldots\vee\varphi_n)\in\text{mon-con}(\alpha)$. Note that it might be that $n=0$ in which case $\varphi_1\vee\ldots\vee\varphi_n$ is the empty disjunction which is $\bot$ and we set $ST_x(\bot)$ to be $\bot$. But then $\alpha\models\bot$ so $\alpha\equiv\bot$.

Since $M,s\models\text{mon-con}(\alpha)$ so $M,s\models ST_x(\varphi_i)$ for some $1\leq i\leq n$. On the other hand, we had $M,s\not\models ST_x(\varphi_i)$ since we supposed that $ST_x(\varphi_i)\not\in\text{mon-tp}_{M}(s)$; a contradiction. We have shown that $\Phi$ is a satisfiable first order theory in a single free variable $x$, so there exists a first order structure $N$ and an element $u\in dom(N)$ such that $N\models\beta(u)$ for all $\beta\in\Phi$. This induces a pointed Kripke model that we also denote by $(N,u)$. 

Let $\text{mon-tp}_M(s):=\{ST_x(\psi)\mid \psi\in\mono\;\&\;M,s\models \psi\}$ \revision{(where ``$\text{mon-tp}$'' stands for monotone type)}. We show that $\text{mon-tp}_N(u)\subseteq\text{mon-tp}_M(s)$. \revision{We argue via contraposition, so consider some $\chi\in\mono$ such that $ST_x(\chi)\not\in\text{mon-tp}_M(s)$. This means that $M,s\models\neg ST_x(\chi)$, so by definition $\neg ST_x(\chi)\in\Phi$. But we know that $N,u$ satisfies the whole theory $\Phi$ so in particular $N,u\models\neg ST_x(\chi)$ whence $\chi\not\in\text{mon-tp}_N(u)$. By Lemma~\ref{lem:omegasaturated} there exists modally saturated pointed Kripke models $(M^*,s^*)$ and $(N^*,u^*)$ satisfying the same modal formulas as $(M,s)$ and $(N,u)$ respectively.} Define the relation:
\[Z:=\{(v,t)\in dom(N^*)\times dom(M^*)\mid \text{mon-tp}_{N^*}(v)\subseteq \text{mon-tp}_{M^*}(t)\}\]
We show that $Z:N^*,u^*\simu M^*,s^*$. First of all, observe that $\text{mon-tp}_{N^*}(u^*)=\text{mon-tp}_{N}(u)\subseteq\text{mon-tp}_{M}(s)=\text{mon-tp}_{M^*}(s^*)$ and thus $(u^*,s^*)\in Z$. 

[atom] This holds trivially since $N^*,v\models p$ means $ST_x(p)\in\text{mon-tp}_{N^*}(v)\subseteq\text{mon-tp}_{M^*}(t)$ so that $M^*,t\models p$ as well. 

[forth'] Let $R^{N^*}vv'$. We consider two cases: (i) $\text{mon-tp}_{N^*}(v')=\emptyset$ or (ii) $\text{mon-tp}_{N^*}(v')\ne\emptyset$. In the latter case (ii), it is easy to see that every finite subset of $\text{mon-tp}_{N^*}(v')$ is satisfiable at some successor of $t$ in $M^*$. For let $ST_x(\varphi_1),\ldots ST_x(\varphi_n)\in\text{mon-tp}_{N^*}(v')$. Then $N,v\models ST_x(\Diamond(\varphi_1\wedge\ldots\wedge\varphi_n))$ and thus $M^*,t\models ST_x(\Diamond(\varphi_1\wedge\ldots\wedge\varphi_n))$ since $(v,t)\in Z$ and thus $\text{mon-tp}_{N^*}(v)\subseteq\text{mon-tp}_{M^*}(t)$. Since $M^*$ is modally saturated, there must be some successor $t'\in R^{M^*}[t]$ that satisfies $\text{mon-tp}_{N^*}(v')$, so $\text{mon-tp}_{N^*}(v')\subseteq\text{mon-tp}_{M^*}(t')$ and therefore $(v',t')\in Z$. If (i) then by Lemma \ref{lem:emptytypemeansemptyloop} we know that $N^*,v'\bisim\emptyloop$.

[back'] \revision{Let $R^{M^*}tt'$. Define \rrevision{$\text{ant-tp}_{M}(s):=\{\neg ST_x(\psi)\mid \psi\in\mono,\;M,s\models\neg\psi\}$}. We will consider again two cases: (i) $\text{ant-tp}_{M^*}(t')=\emptyset$ or (ii) $\text{ant-tp}_{M^*}(t')\ne\emptyset$. In the latter case (ii), it is easy to see that every finite subset of $\text{ant-tp}_{M^*}(t')$ is satisfied at some successor of $v$ in $N^*$. For let $\neg ST_x(\varphi_1),\ldots, \neg ST_x(\varphi_n)\in\text{ant-tp}_{M^*}(t')$. It follows that $M^*,t\models ST_x(\Diamond(\neg\varphi_1\wedge\ldots\wedge\neg\varphi_n))$ so as $\neg ST_x(\Diamond(\neg\varphi_1\wedge\ldots\wedge\neg\varphi_n))\equiv\Box(\varphi_1\vee\ldots\vee\varphi_n)$ and the latter formula is in $\mono$ we have that $ST_x(\Box(\varphi_1\vee\ldots\vee\varphi_n))\not\in\text{mon-tp}_{M^*}(t)$. But as $(v,t)\in Z$ we have $\text{mon-tp}_{N^*}(v)\subseteq\text{mon-tp}_{M^*}(t)$ so also $ST_x(\Box(\varphi_1\vee\ldots\vee\varphi_n))\not\in\text{mon-tp}_{M^*}(v)$. By the same reasoning again we get $N^*,v\models ST_x(\Diamond(\neg\varphi_1\wedge\ldots\wedge\neg\varphi_n))$ which means that $\{\neg ST_x(\varphi_1),\ldots,\neg ST_x(\varphi_n)\}$ is satisfiable at some successor of $v$ in $N^*$. Hence, by modal saturation of $N^*$ there exists some successor $v'\in R^{N^*}[v]$ that satisfies $\text{ant-tp}_{M^*}(t')$. Finally, we show that $\text{mon-tp}_{N^*}(v')\subseteq\text{mon-tp}_{M^*}(t')$, which shows that $(v',t')\in Z$. We argue by contraposition, so let $ST_x(\varphi)\not\in\text{mon-tp}_{M^*}(t')$. By definition of $\text{ant-tp}$ we get $\neg ST_x(\varphi)\in\text{ant-tp}_{M^*}(t')$, and since $N^*,v'$ satisfies $\text{ant-tp}_{M^*}(t')$ it follows that $\neg ST_x(\varphi)\in\text{ant-tp}_{N^*}(v')$ as well. But then by definition of $\text{ant-tp}$ again we have that $ST_x(\varphi)\in\text{mon-tp}_{N^*}(v')$.}

\rrevision{
We conclude by chasing the diagram: 
by construction, $N,u\models\alpha$. Hence,
since $(N,u)$ and $(N^*,u^*)$ agree on all
first-order formulas in one free variable, $N^*,u^*\models\alpha$.
Therefore, by preservation under weak simulations,
$M^*,s^*\models\alpha$, and therefore,
since $(M,s)$ and $(M^*,s^*)$ agree on all
first-order formulas in one free variable, $M,s\models\alpha$.}
\end{proof}

\begin{remark}
Rosen~\cite{Rosen1997:modal} showed that the classic Van Benthem characterization of modal logic in terms of
bisimulation-preservation holds also in the finite, i.e., when restricted to finite structures.
We leave it as an open problem whether Theorem~\ref{thm:wsim-characterization} holds in the finite.
\end{remark}

\begin{remark}
An immediate corollary  of Lemma~\ref{lem:emptyloopsatall} is that the syntactic restriction to $\mono$ thus trivialises the problems of determining satisfiability and validity. That is, all $\mono$ formulas are satisfiable, and none of them is valid. Although this makes the satisfiability problem for $\mono$ trivial, checking entailment for $\mono$ is still PSPACE-complete like the corresponding problem for the full modal language. \revision{This can be shown by adapting an argument from \cite[Theorem 2.2.6 on p.~38]{kerdiles} to the modal case,
which involves substituting fresh positive propositional variables for negated atoms (after bringing the formula in negation normal form), and rewriting the resulting formula accordingly so as not to contain negation.} This is one instance of the more general phenomenon that more fine-grained structure of traditional reasoning problems emerges from the study of language-fragments.
\end{remark}

\subsection{The Category of Pointed Models and Weak Simulations}
\label{sec:category}

In this section, we study weak simulations from 
a category-theoretic perspective. We start by 
showing that pointed models \rrevision{with} 
weak simulations indeed form a category. Recall 
that a category consists of a collection of \emph{objects} and a collection of 
\emph{morphisms}, together a \emph{composition} operation
on morphisms that satisfies suitable axioms (namely,
associativity and the existence of ``identity'' morphisms). In our case, the objects are pointed models and the morphisms are weak simulations.
Furthermore, our composition operation will be ordinary \emph{relational composition}.
Thus, we need to show that the relational composition of two weak simulations is a weak simulation, and we need to show that every pointed model has a weak simulation to itself that acts as an identity morphism.

The following lemma will help us to prove closure under relational composition.


\begin{lemma}\label{lem:bisimemptyloop} ~
\begin{itemize}
    \item If $M,s\simu M',s'$ and $M',s'\bisim\emptyloop$, then $M,s\bisim\emptyloop$. 
    \item If $M,s\bisim\fullloop$ and $M,s\simu M',s'$ then $M',s'\bisim\fullloop$.
\end{itemize}
\end{lemma}
\begin{proof}
Let $Z:M,s\simu M',s'$ and $Z':M',s'\bisim\emptyloop$. Without loss of generality (i.e. up to bisimilarity), assume that $M=M\restriction s$ and $M'=M'\restriction s'$. By Lemma~\ref{lem:pointgeneratedbisimtotal} it must be that $Z'$ is total, i.e. $Z'=dom(M')\times dom(\emptyloop)$. Let $Z^*:=dom(M)\times dom(\emptyloop)$; we show that this is a bisimulation between $(M,s)$ and $\emptyloop$. So let $(t,r)\in Z^*$ be arbitrary. 

[atom] Since $M=M\restriction s$, there is a path $(t_0,s\ldots,t_n)$ in $M$ with $t_0=s$ and $t_n=t$. Using the [forth] clause of $Z$ repeatedly, it follows that there is a path $(t'_0,\ldots,t'_n)$ in $M'$ with $t_0=t$ and $(t_i,t'_i)\in Z$ for all $i\leq n$. Also, the fact that $Z'$ is a bisimulation it follows that $v^{M'}(t')=\emptyset$ for all $t'\in dom(M')$, so in particular $v^{M'}(t'_n)=\emptyset$. But as $Z$ is a weak simulation $v^M(t)\subseteq v^{M'}(t'_n)$ so $v^{M}=\emptyset$ as well.

[forth] If $R^Mtt'$ then clearly $(t',r)\in Z^*$ by definition and $r$ sees itself. 

[back] Note that $r$ only sees itself so $Rrr'$ implies $r'=r$. We show that there is some successor $t'$ of $t$ such that $(t',r)\in Z^*$. By definition of $Z^*$, any successor would suffice. But note that $M\restriction t,t\bisim M,t\bisim\emptyloop$ hence $t$ must have some successor.

\rrevision{The proof of the second item is omitted as it is analogous.}
\end{proof}

\begin{theorem}
Weak simulations are closed under relational composition: if $Z_1:M_0,s_0\simu M_1,s_1$ and $Z_2:M_1,s_1\simu M_2,s_2$,
then $Z_1\circ Z_2:M_0,s_0\simu M_2,s_2$.
\end{theorem}
\begin{proof}
Let $(t_0,t_2)\in Z_1\circ Z_2$. Then there is some $t_1$ with $(t_0,t_1)\in Z_1$ and $(t_1,t_2)\in Z_2$. 

[atom] Clearly $v(t_0)\subseteq v(t_1)\subseteq v(t_2)$ by the [atom] clauses of $Z_1$ and $Z_2$. 

[forth'] If $Rt_0u_0$, by [forth'] of $Z_1$ either $M_0,u_0\bisim\emptyloop$ or there is a matching successor $u_1$ of $t_1$ with $(u_0,u_1)\in Z_1$. In the former case, we can use the escape clause for the pair $(t_0,t_2)$ and we are done. Else, we apply the [forth'] clause on $u_1$ and get that either $M_1,u_1\bisim\emptyloop$ or there is a successor $u_2$ of $t_2$ with $(u_1,u_2)\in Z_2$ and hence $(u_0,u_2)\in Z_1\circ Z_2$. In the latter case, we would be done. In the former case we have $M_0,u_0\simu M_1,u_1\bisim\emptyloop$ so by lemma \ref{lem:bisimemptyloop} $M_0,u_0\bisim\emptyloop$ and we can use the escape clause on $(t_0,t_2)\in Z_1\circ Z_2$. 

[back'] If $Rt_2u_2$ by [back'] of $Z_2$ either $M_2,u_2\bisim\fullloop$ or there is a matching successor $u_1$ of $t_1$ with $(u_1,u_2)\in Z_2$. In the former case, we can use the escape clause for the pair $(t_0,t_2)$ and we are done. Else, we apply the [back'] clause on $u_1$ and get that either $M_1,u_1\bisim\fullloop$ or there is a successor $u_0$ of $t_0$ with $(u_0,u_1)\in Z_1$ and hence $(u_0,u_2)\in Z_1\circ Z_2$. In the latter case, we would be done. In the former case we have $\fullloop\bisim M_1,u_1\simu M_2,u_2$ so by lemma \ref{lem:bisimemptyloop} $M_2,u_2\bisim\fullloop$ and we can use the escape clause on $(t_0,t_2)\in Z_1\circ Z_2$.
\end{proof}

Since our composition operation is ordinary
relational composition, it is clearly
an associative operation. 
%
%
%
It rests therefore only to observe that the ``diagonal'' relation $1_{(M,s)}=\{(t,t)\mid t\in dom(M)\}$ is indeed a weak simulation that acts as an identity morphism, for each pointed model $(M,s)$.

\begin{corollary}
The class of pointed models, with weak simulations as morphisms, forms a category $\mathbf{wSim}[\Prop]$.
\end{corollary}

Next, we will study the basic category-theoretic properties of the category $\mathbf{wSim}[\Prop]$. A weakly initial object in a category is a not necessarily unique (up to isomorphism) object that has a morphism into every object in that category. Dually, a weakly final object is a not necessarily unique (up to isomorphism) object such that every object in the category has a morphism into this object.

\begin{theorem}
The category $\mathbf{wSim}[\Prop]$ has a weakly initial object, namely $\emptyloop$ and a weakly final object, namely $\fullloop$.
\end{theorem}

\begin{proof}
Follows immediately from Lemma~\ref{lem:Winitialemptyloop}.
Note that there is not a \textit{unique} weak simulation from $\emptyloop$ to other pointed models $(M,s)$, not even when $M$ is point-generated. A counterexample is when $M=(\{s,t\},\{(s,t),(t,t)\},v)$ where $v(s)=v(t)=\mathrm{Prop}\}$. Then both $\{(r,s)\}$ and $\{(r,s),(r,t)\}$ are weak simulations $\emptyloop\simu M,s$ (where $dom(\emptyloop)=\{r\}$ is the root of the empty loop model).
\end{proof}

The next natural question that arises is whether our category has products and coproducts. A category has products if for every two objects $A$ and $B$ there is an object $A\times B$ with morphisms $\pi_0: A\times B\to A, \pi_1: A\times B\to B$ (the projections onto the left- and right coordinate, respectively) such that for every object $C$ and pairs of morphisms $h:C\to A,f:C\to B$ there is a unique morphism $g:C\to A\times B$ such that $\pi_0\circ g=h$ and $\pi_1\circ g=f$. Equivalently, \revision{there is a unique $g$ for which the following diagram commutes:}

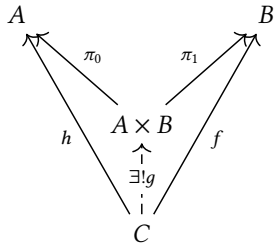
\begin{center}
\begin{tikzcd}[row sep=large]
A &                                                                                 & B \\
  & A\times B \arrow[lu, "\pi_0"'] \arrow[ru, "\pi_1"]                              &   \\
  & C \arrow[luu, "h"] \arrow[ruu, "f"'] \arrow[u, "\exists!g" description, dashed] &  
\end{tikzcd}
\end{center}

Dually, a category has coproducts if for every two objects $A$ and $B$ there is an object $A\sqcup B$ with morphisms $\rho_0: A\to A\sqcup B, \rho_1:B\to A\sqcup B$ (the co-projections) such that for every object $C$ and pairs of morphisms $h:A\to C,f:B\to C$ there is a unique morphism $g:A\sqcup B\to C$ such that $g\circ\rho_0=h$ and $g\circ\rho_1=f$. Equivalently, $g$ must make the following diagram commute:

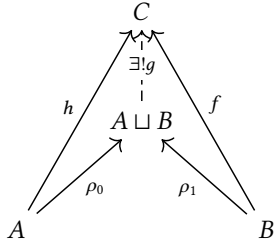
\begin{center}
\begin{tikzcd}[row sep=large]
                                         & C                                                    &                                          \\
                                         & A\sqcup B \arrow[u, "\exists!g" description, dashed] &                                          \\
A \arrow[ru, "\rho_0"'] \arrow[ruu, "h"] &                                                      & B \arrow[lu, "\rho_1"] \arrow[luu, "f"']
\end{tikzcd}
\end{center}

\begin{theorem}\label{thm:nocoproducts}
The category $\mathbf{wSim}[\Prop]$ does not have coproducts.
\end{theorem}
\begin{proof}
It suffices to show that there are two pointed models $(E,e),(E',e')$ such that their coproduct $(E,e)\sqcup(E',e')$ does not exist. Consider the following two models depicted below;
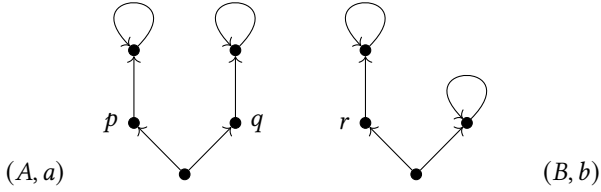
\begin{center}
{ $(A,a)$}
\begin{tikzpicture}[scale = 0.1, node distance=0.8cm]
\node[point] (root) [] {};
\node[point] (v) [above left = of root, label=left:{ $p$}] {};
\node[point] (u) [above right = of root,label = right:{ $q$}] {};
\node[point] (s) [above = of v] {};
\node[point] (t) [above = of u] {};

\path[->] (root) edge (v);
\path[->] (root) edge (u);
\path[->] (v) edge (s);
\path[->] (u) edge (t);
\path[->] (t) edge[reflexive point] (t);
\path[->] (s) edge[reflexive point] (s);
\end{tikzpicture}
\begin{tikzpicture}[scale = 0.1, node distance=0.8cm]
\node[point] (root) [] {};
\node[point] (v) [above left = of root, label=left:{ $r$}] {};
\node[point] (u) [above right = of root] {};
\node[point] (s) [above = of v] {};

\path[->] (root) edge (v);
\path[->] (root) edge (u);
\path[->] (v) edge (s);
\path[->] (u) edge[reflexive point] (u);
\path[->] (s) edge[reflexive point] (s);
\end{tikzpicture}
{ $(B,b)$}
\end{center}
Note that $A,a\models\Box(p\vee q)$ and $B,b\models\Diamond r$. We will show that $(A,a)\sqcup(B,b)$ does not exist. 

Suppose for contradiction that $(A,a)\sqcup(B,b)$ does exist. Then by the properties of coproducts it must be that $(A,a)\simu(A,a)\sqcup(B,b)$ and $(B,b)\simu (A,a)\sqcup(B,b)$ and that for any $(C,c)$ for which $(A,a)\simu (C,c)$ and $(B,b)\simu(C,c)$ we have $(A,a)\sqcup(B,b)\simu(C,c)$ (universal property). By preservation, it immediately follows that $(A,a)\sqcup(B,b)\models\Box(p\vee q)\wedge\Diamond r$. However, observe that $\Box(p\vee q)\wedge\Diamond r\models\Diamond(p\wedge r)\vee\Diamond(q\wedge r)$. Hence either (i) $(A,a)\sqcup(B,b)\models\Diamond(p\wedge r)$ or (ii) $(A,a)\sqcup(B,b)\models\Diamond(q\wedge r)$. Now consider the following two models $(C,c)$ and $(C',c')$, both of which are `above' both $(A,a)$ and $(B,b)$ in the weak simulation order.
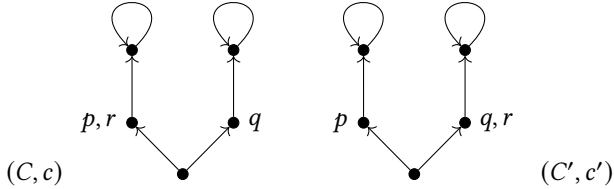
\begin{center}
{ $(C,c)$}
\begin{tikzpicture}[scale = 0.1, node distance=0.8cm]
\node[point] (root) [] {};
\node[point] (v) [above left = of root, label=left:{ $p,r$}] {};
\node[point] (u) [above right = of root,label = right:{ $q$}] {};
\node[point] (s) [above = of v] {};
\node[point] (t) [above = of u] {};

\path[->] (root) edge (v);
\path[->] (root) edge (u);
\path[->] (v) edge (s);
\path[->] (u) edge (t);
\path[->] (t) edge[reflexive point] (t);
\path[->] (s) edge[reflexive point] (s);
\end{tikzpicture}
\begin{tikzpicture}[scale = 0.1, node distance=0.8cm]
\node[point] (root) [] {};
\node[point] (v) [above left = of root, label=left:{ $p$}] {};
\node[point] (u) [above right = of root,label = right:{ $q,r$}] {};
\node[point] (s) [above = of v] {};
\node[point] (t) [above = of u] {};

\path[->] (root) edge (v);
\path[->] (root) edge (u);
\path[->] (v) edge (s);
\path[->] (u) edge (t);
\path[->] (t) edge[reflexive point] (t);
\path[->] (s) edge[reflexive point] (s);
\end{tikzpicture}
{ $(C',c')$}
\end{center}
It is easily verified that indeed $(A,a)\simu (C,c)$,$(B,b)\simu(C,c)$ and $(A,a)\simu (C',c')$, $(B,b)\simu(C',c')$. We will not spell out these weak simulations in detail. Moreover, observe that $C,c\not\models\Diamond(q\wedge r)$ and that $C',c'\not\models\Diamond(p\wedge r)$. However, it follows from the universal property of $(A,a)\sqcup(B,b)$ that $(A,a)\sqcup(B,b)\simu(C,c)$ and $(A,a)\sqcup(B,b)\simu(C',c')$, i.e. the coproduct $(A,a)\sqcup(B,b)$ must be `below' both $(C,c)$ and $(C',c')$ in the weak simulation order. This is a contradiction, since if (i) holds then $C',c'\models\Diamond(p\wedge r)$ by preservation and if (ii) then $C,c\models\Diamond(q\wedge r)$ by preservation. Thus the coproduct $(A,a)\sqcup(B,b)$ of $A,a$ and $B,b$ cannot exist.
\end{proof}

As we will see soon, 
$\mathbf{wSim}[\Prop]$ does not have products
either. Rather than giving a concrete counterexample,
we will show that $\mathbf{wSim}[\Prop]$ 
is in fact isomorphic to its own opposite category, and we will rely on the fact that products in a category correspond to co-products in the opposite category. More precisely,
consider the following `flipping' operation on models.

\begin{definition}\label{def:negonmodels}
Given a model $M=(dom(M),R,v)$, let  $M^{\neg}=(dom(M),R,v^{\neg})$ be obtained from $M$ by `flipping' all the valuations, i.e.~$v^{\neg}(t)=\mathrm{Prop}-v(t)$ for all $t\in dom(M)$. 
\end{definition}

In what follows we use $Z^{-1}$ to denote
the relational converse of $Z$, i.e., 
$Z^{-1}=\{(b,a) \mid (a,b)\in Z\}$.

\begin{proposition}\label{prop:flippingsimus}
Whenever $Z:E,e\simu M,s$ then $Z^{-1}:M^{\neg},s\simu E^{\neg},e$
\end{proposition}
\begin{proof}
Let $(t,f)\in Z^{-1}$ be arbitrary, then $(f,t)\in Z$. [atom] By the [atom] clause of $Z$ we get $v(f)\subseteq v(t)$ and hence $v^{\neg}(t)\subseteq v^{\neg}(f)$. [forth'] If $Rtu$ and $M^{\neg},u\;\underline{\not\leftrightarrow}\;\emptyloop$ then $M,u\;\underline{\not\leftrightarrow}\;\fullloop$ and hence by the [back'] clause of $Z$ there must be some successor $g$ of $f$ with $(g,u)\in Z$ and hence $(u,g)\in Z^{-1}$. [back'] Similarly, if $Rfg$ and $E^{\neg},g\;\underline{\not\leftrightarrow}\;\fullloop$ then $E,g\;\underline{\not\leftrightarrow}\;\emptyloop$ so by [forth'] of $Z$ there must be some successor $u$ of $t$ with $(g,u)\in Z$ and hence $(u,g)\in Z^{-1}$.
\end{proof}

In light of Proposition~\ref{prop:flippingsimus}, we can naturally view the operation $(\cdot)^{\neg}$ as a ``contravariant endo-functor'' on $\mathbf{wSim}[\Prop]$.
Formally, let $\mathbf{wSim}[\Prop]^{op}$ be the 
category that is identical to $\mathbf{wSim}[\Prop]$,
except that the direction of the morphisms is 
reversed. We define $(\cdot)^{\neg}:\mathbf{wSim}[\Prop]\to\mathbf{wSim}[\Prop]^{op}$
as the functor obtained by setting $(M,s)^{\neg}:=(M^{\neg},s)$ on objects and $Z^{\neg}:=Z^{-1}$ on morphisms. It is easy to see that this is indeed a functor as it commutes with composition and identity morphisms. 

\begin{proposition}\label{prop:isomorphism}
The functor $(\cdot)^{\neg}:\mathbf{wSim}[\Prop]\to\mathbf{wSim}[\Prop]^{op}$ is an isomorphism.
\end{proposition}
\begin{proof}
\revision{Observe that the functor $(\cdot)^{\neg}:\mathbf{wSim}[\Prop]\to\mathbf{wSim}[\Prop]^{op}$ induces an inverse, namely the \emph{opposite functor} $(\cdot)^{\neg^{op}}:\mathbf{wSim}[\Prop]^{op}\to\mathbf{wSim}[\Prop]$ defined by $(M,s)^{\neg^{op}}:=(M^{\neg},s)$ and $Z^{\neg^{op}}:=Z^{-1}$. That is, $(\cdot)^{\neg^{op}}$ is defined just like $(\cdot)^{\neg}$ except that the domain and codomain have switched.} We have to show that the composition of functors $(\cdot)^{\neg}\circ(\cdot)^{\neg{op}}$ is the same as the identity functor on $\mathbf{wSim}$. But we know that $((M,s)^{\neg})^{\neg{op}}=(M^{\neg\neg},s)=(M,s)$ and clearly also $(Z^{\neg})^{\neg^{op}}=(Z^{-1})^{-1}=Z$.
\end{proof}

Since products in $\mathbf{wSim}[\Prop]$ correspond to co-products in $\mathbf{wSim}[\Prop]^{op}$, 
it follows that $\mathbf{wSim}$ has products iff it has coproducts. Therefore, we obtain:

\begin{corollary}
The category $\mathbf{wSim}[\Prop]$ does not have products.
\end{corollary}


Finally, we show that our isomorphism also coheres nicely with the following syntactic operation on formulas. We will make crucial use of this observation in the proof of the main theorem in the following section. We have already dealt with these operations more informally in the proof of Theorem \ref{thm:dualposex}, here we make these observations more explicit.

\begin{definition}\label{def:negonformulas}
Given a modal formula $\varphi$ with variables among $\Prop=\{p_1,\ldots,p_n\}$, let $\varphi^{\neg}:=\varphi[\neg p_1/p_1]\;\ldots\;[\neg p_n/p_n]$. Note that $\varphi^{\neg\neg}\equiv\varphi$, i.e.~our operation is idempotent up to logical equivalence (by removing double negations).
\end{definition}

\begin{proposition}\label{prop:coherencenegoperations}
For all formulas $\varphi$ and pointed models $(M,s)$
\[M,s\models\varphi\qquad\textrm{iff}\qquad M^{\neg},s\models\varphi^{\neg}\]
\end{proposition}
\begin{proof}
By induction on the complexity of $\varphi$. The Boolean case are immediate by inductive hypothesis. The cases for the modal operators $\Box$ and $\Diamond$ are also immediate by inductive hypothesis since $(\Box\psi)^{\neg}=\Box(\psi^{\neg})$ and $(\Diamond\psi)^{\neg}=\Diamond(\psi^{\neg})$. The atomic case follows easily from the definitions as $M,s\models p$ iff $p\in v(s)$ iff $p\not\in v^{\neg}(s)$ iff $M^{\neg},s\models\neg p$.
\end{proof}

\subsection{Construction of Weak-Simulation Dualities}
We first define what weak simulations dualities are in the context of $\mathbf{wSim}[\Prop]$, and show that they correspond to finite characterisations of $\mono$-formulas.

\begin{definition}
A \emph{weak simulation duality} is a pair $(\mathcal{D}_{pos},\mathcal{D}_{neg})$ of finite sets of pointed models such that every pointed model either weakly simulates some pointed model from $\mathcal{D}_{pos}$, or is weakly simulated by some pointed model from $\mathcal{D}_{neg}$, where the ``or'' is meant exclusively. That is, for $\upclosure\mathcal{D}=\{(E,e)\mid (E',e')\simu_w(E,e)$ for some $(E',e')\in\mathcal{D}\}$ and   $\downclosure\mathcal{D}=\{(E,e)\mid (E,e)\simu_w(E',e')\;\text{for some}(E',e')\in\mathcal{D}\}$ we have that $\upclosure\mathcal{D}_{pos}$ and $\downclosure\mathcal{D}_{neg}$ partition $\mathbf{wSim}[\Prop]$.
\end{definition}

\begin{theorem}\label{thm:dualitiesarecharacterisations}
Let $\mathbb{E}=(\mathcal{D}_{pos}, \mathcal{D}_{neg})$ be a weak simulation duality where the examples in $\mathcal{D}_{pos}$ have been labelled positively and the examples in $\mathcal{D}_{neg}$ negatively. If some $\varphi\in\mono$ fits $\mathbb{E}$ then $\mathbb{E}$ is a finite characterisation of $\varphi$ w.r.t. $\mono$.
\end{theorem}
\begin{proof}
Let $\varphi,\psi\in\mono$ and $\mathbb{E}=(E^+,E^-)$ be a weak-simulation duality such that $\varphi$ and $\psi$ both fit these labelled examples. To see that $\varphi\models\psi$, let $M,s\models\varphi$. By  preservation (Theorem \ref{thm:wSimPreservation}), it cannot be that $(M,s)\in\downclosure E^-$ otherwise $\varphi$ would be satisfied on, and therefore fail to fit \rrevision{as a negative example}, some example in $E^-$. Hence, since $\upclosure E^+$ and $\downclosure E^-$ partition $\mathbf{wSim}[\Prop]$, it follows that $(M,s)\in\upclosure E^+$ but then $M,s\models\psi$ as well by preservation and the fact that $\psi$ fits $E^+$ \rrevision{as positive examples}. The argument that $\psi\models\varphi$ is completely symmetric.
\end{proof}

By the above theorem, it suffices to construct on input $\varphi\in\mono$, a weak simulation duality $(E^+_{\varphi},E^-_{\varphi})$ that $\varphi$ fits. We will now give a direct construction of $E^+_{\varphi}$ from $\varphi$, going through a sort of modal version of the disjunctive normal form developed by Kit Fine \cite{Finenormalform}; computing the examples compositionally from arbitrary formulas would be hard given the complex role of conjunction as a sort of `merging' or `gluing' operation on models. Then, we continue to show how this definition of $E^+_{\varphi}$ already implicitly defines also the negative examples $E^-_{\varphi}$ via the automorphism $(\cdot)^{\neg}$ of $\mathbf{wSim}[\Prop]$.

\begin{definition}{(Fine Normal Form)}
A \textit{basic normal form} of level $0$ is a non-empty conjunction of \textit{positive} literals. A basic normal form of level $n+1$ is a non-empty conjunction of formulas
\[\pi\wedge\Diamond\varphi_0\wedge\ldots\wedge\Diamond\varphi_n\wedge\Box(\psi_0\vee\ldots\vee\psi_m)\]
where $\pi$ is a (possibly empty) conjunction of positive (i.e.~unnegated) propositional variables $p$, and each $\varphi_i,\psi_j$ is a basic normal form of level at most $n$ (this ensures that a normal form of level $n$ is also a normal form of level $n'$ for all $n'>n$). A \textit{normal form} of level $m$ is a non-empty disjunction of basic normal forms of level $m$.
\end{definition}

Note that we explicitly require the conjunctions and disjunctions to be non-empty in the above definition, in order to exclude $\top$ and $\bot$. Intuitively, basic normal forms are modal formulas in negation normal form that may have disjunctions under the scope of boxes but no disjunctions under the scope of any diamonds.

\begin{theorem}\label{thm:FineNormalForm}For every formula $\xi\in\mono$, there is a normal form $nf(\xi)$ of level $d(\xi)$ such that $\xi\equiv nf(\xi)$.
\end{theorem}
\begin{proof}
By induction on $d(\xi)$. If $d(\xi)=0$ we can just use the disjunctive normal form for propositional logic, i.e.~we simply distribute $\wedge$'s over $\vee$'s via the validity $p\wedge(q\vee r)\equiv (p\wedge q)\vee(p\wedge r)$. If $d(\xi)=n+1$, observe that $\xi$ is a Boolean combination of formulas $\Diamond\varphi$, $\Box\psi$ and propositional variables $p$, where each $d(\varphi),d(\psi)\leq n$. By the inductive hypothesis, we may assume that each such $\varphi$ and $\psi$ are normal forms of level $n$. This means that $\varphi\equiv\bigvee\Phi$ and $\psi\equiv\bigvee\Psi$ for some nonempty finite sets $\Phi,\Psi$ of basic normal forms of level $n$. By Kripke normality we transform $\Diamond\varphi=\Diamond\bigvee\Phi\equiv\bigvee_{\varphi\in\Phi}\Diamond\varphi$, after which $\xi$ is equivalent to a Boolean combination of formulas of the form $\Diamond\varphi$, formulas of the form \revision{$\Box(\psi_0\vee\ldots\vee\psi_m)$} and positive propositional variables $p$, where $\varphi,\psi_0,\ldots,\psi_m$ are all basic normal forms of level $n$.

\revision{Applying the propositional distributive law again we obtain the disjunctive normal form of this Boolean combination, which is a big disjunction over conjunctions of formulas of the form $\Diamond\varphi, \Box(\psi_0\vee\ldots\vee\psi_m)$ and positive propositional variables $p$, where  $\varphi,\psi_0,\ldots,\psi_m$ are all basic normal forms of level $n$). We apply Kripke normality in the form of the equivalence $\bigwedge_i\Box\varphi_i\equiv\Box(\bigwedge_i\varphi_i)$ to conjunctions of formulas of the form $\Box(\psi_0\vee\ldots\vee\psi_m)$, to rewrite then into one formula starting with a $\Box$ and under its scope a conjunction of disjunctions of basic normal forms of level $n$. Applying the distributive law again we rewrite this into a disjunction of conjunctions of basic normal forms of level $n$. Finally then, the whole formula is rewritten to a disjunction of conjunctions of the form $\pi\wedge\Diamond_0\wedge\ldots\wedge\Diamond\varphi_n\wedge\Box(\psi_0\vee\ldots\psi_m)$ where $\varphi,\psi_0,\ldots,\psi_m$ are all basic normal forms of level $n$. That is, the whole formula is equivalently rewritten into a disjunction of basic normal forms of level $n+1$.}
\end{proof}

The following corollary is thus immediate from the fact that every formula in $\mono$ can be brought into the above Fine normal form.

\begin{corollary}\label{cor:cases}
Every basic normal form of level $k$ is of one of the following forms:
\begin{align*}
    &(i)\;\;\;p_1\wedge\ldots\wedge p_n       &\textrm{if}\;n=0\\
    &(ii)\;\;\pi\wedge\Diamond\varphi_0\wedge\ldots\wedge\Diamond\varphi_n   &\textrm{if}\;n>0\\
    &(iii)\;\pi\wedge\Box(\psi_0\vee\ldots\vee\psi_m)
    &\textrm{if}\;n>0\\
    &(iv)\;\;\pi\wedge\Diamond\varphi_0\wedge\ldots\wedge\Diamond\varphi_n\wedge\Box(\psi_0\vee\ldots\vee\psi_m) &\textrm{if}\;n>0
\end{align*}
where each $\varphi_i,\psi_j$ is a basic normal form of level $k-1$ and $\pi$ is a (possibly empty conjunction) of positive atoms.
\end{corollary}

Now comes the direct construction of $E^+_{\varphi}$ given $\varphi\in\mono$ in Fine normal form. Define the following  `\textit{gluing}' operator (which we denote by $\blacktriangledown$ to suggest the connection with the syntactic $\nabla$-operator \cite{MOSS1999277}) for building the positive examples.

\begin{definition}
Given a set of pointed models $\mathbb{E}$ and a set of propositional variables $P\subseteq\mathrm{Prop}$ let $\blacktriangledown_P(\mathbb{E})$ be the model obtained by gluing all the examples in $\mathbb{E}$ to a new root $r$ with $v(r)=P$. That is:
\begin{itemize}
    \item Domain $dom(\blacktriangledown_P(\mathbb{E}))=\biguplus_{(E,e)\in\mathbb{E}}\;dom(E)\uplus\{r\}$
    \item Relation $R^{\blacktriangledown}[r]=\{e\mid (E,e)\in\mathbb{E}\}$ and $R^{\blacktriangledown}$ $\restriction dom(E)=R^E$ for all $(E,e)\in\mathbb{E}$
    \item Colouring $v^{\blacktriangledown}\restriction dom(E)=v^E$ for all $(E,e)\in\mathbb{E}$ and $v^{\blacktriangledown}(r)=P$
\end{itemize}

\end{definition}
We write $\blacktriangledown_P(E,e)$ rather than $\blacktriangledown_P(\{(E,e)\})$ for singletons. By Theorem \ref{thm:FineNormalForm}, every $\varphi\in\mono$ is equivalent to a normal form $nf(\varphi)$ of level $d(\varphi)$. In fact, we observed there (Corollary \ref{cor:cases}) that every \rrevision{basic} normal form is of one of the forms (i)-(iv). Hence, we will define the positive examples for these four cases separately. Note that for a conjunction of literals $\pi$, $\hat{\pi}$ denotes the set of propositional variables occurring in $\pi$. Further, we suppress the distinguished element for readability, writing $E$ instead of $(E,e)$.

\begin{definition}\label{def:posexamples}
\revision{For every formula $\varphi$ in $\mono$ in Fine normal form, we define the positive examples of $\varphi$ to be the positive examples for $nf(\varphi)$, in symbols $E^+_{\varphi}:=E^+_{nf(\varphi)}$, where the latter is defined as follows:}
\begin{align*}
    E^+_{\pi}&:=\{\blacktriangledown_{\hat{\pi}}(\emptyloop)\}\\
    E^+_{\pi\wedge\Diamond\varphi_0\wedge\ldots\wedge\Diamond\varphi_n}&:=\{\blacktriangledown_{\hat{\pi}}(\{E_1,\ldots,E_n\}\cup\{\emptyloop\})\mid E_i\in E^+_{\varphi_i}\;\textrm{for all}\;i\leq n\}\\
    E^+_{\pi\wedge\Box(\psi_0\vee\ldots\vee\psi_m)}&:=\{\blacktriangledown_{\hat{\pi}}(\mathbb{E})\mid \mathbb{E}\subseteq\bigcup_{j\leq m} E^+_{\psi_j}\}\\
    E^+_{\pi\wedge\Diamond\varphi_0\wedge\ldots\wedge\Diamond\varphi_n\wedge\Box(\psi_0\vee\ldots\vee\psi_m)}&:=\{\blacktriangledown_{\hat{\pi}}(\{E_1,\ldots,E_n\}\cup\mathbb{E})\mid E_i\in\bigcup_{j\leq m} E^+_{\varphi_i\wedge\psi_j}\;\forall i\leq n\;,\;\emptyset\ne\mathbb{E}\subseteq\bigcup_{j\leq m}E^+_{\psi_j}\}\\
    E^+_{\bigvee_{i\leq n}\varphi_i}&:=\bigcup_{i\leq n} E^+_{\varphi_i}
\end{align*}
where $\hat{\pi}$ is the set of all atoms occurring in the conjunction of literals/term $\pi$, and importantly, in the one but last line, $E^+_{\varphi_i\wedge\psi_j}$ is already defined by inductive hypothesis.%
\footnote{\rrevision{Note that $\varphi_i\wedge\psi_j$ is not necessarily in Fine normal form, but may be rewritten
to one without increasing its modal depth.}}
\end{definition}

\begin{proposition}\label{prop:fitsposexamples}
Every $\varphi\in\mono$ fits $E^+_{\varphi}$ \rrevision{as positive examples}.
\end{proposition}
\begin{proof}
By induction on the complexity of normal forms. We show that for each basic normal form $\xi$, $\xi$ fits $E^+_{\xi}$ \rrevision{as positive examples}. From this it follows that every normal form fits its set of positive examples by the properties of union and disjunction. By corollary \ref{cor:cases}, it suffices to consider the cases (i)-(iv). Starting with (i) $\xi= p_1\wedge\ldots\wedge p_n$, clearly $\blacktriangledown_{\{p_1,\ldots,p_n\}}(\emptyloop)\models p_1\wedge\ldots\wedge p_n$. For case (ii) $\xi=\pi\wedge\Diamond\varphi_1\wedge\ldots\wedge\Diamond\varphi_m$, let $E_1,\ldots E_m$ positive examples for $\varphi_1,\ldots,\varphi_m$ respectively. Then $\blacktriangledown_{\pi}(E_1,\ldots,E_m,\emptyloop)\models\xi$ because $\pi$ is satisfied at the root and each conjunct of the form $\Diamond\varphi_i$ is witnessed by some successor whose generated model is precisely some example in $E^+_{\varphi_i}$, which are all models of $\varphi_i$ by inductive hypothesis.

For case (iii) $\xi=\pi\wedge\Box(\psi_0\vee\ldots\vee\psi_m)$, let $\mathbb{E}\subseteq\bigcup_{i\leq m} E^{+}_{\psi_i}= E^+_{\psi_0\vee\ldots\vee\psi_m}$. Thus by inductive hypothesis all examples in $\mathbb{E}$ satisfy $\psi_0\vee\ldots\vee\psi_m$ (note that $\mathbb{E}$ may be empty). But then clearly $\blacktriangledown_{\pi}(\mathbb{E})\models\xi$ as the generated model of every successor of the root of this model is in $\mathbb{E}\subseteq E^{+}_{\psi_0\vee\ldots\vee\psi_m}$. Lastly, for case (iv)  $\xi=\pi\wedge\Box(\psi_0\vee\ldots\vee\psi_m)\wedge\Diamond\varphi_0\wedge\ldots\wedge\Diamond\varphi_n$, let $E_1,\ldots E_m$ be positive examples for $\varphi_1\wedge(\psi_0\vee\ldots\vee\psi_m),\;\ldots,\;\varphi_m\wedge(\psi_0\vee\ldots\vee\psi_m)$ respectively and $\mathbb{E}\subseteq E^{+}_{\psi_0\vee\ldots\vee\psi_m}$. Then $\blacktriangledown_{\hat{\pi}}(\{E_1,\ldots, E_n\}\cup\mathbb{E})\models\xi$ because $\pi$ is satisfied at the root, each conjunct $\Diamond\varphi_i$ is witnessed by some $E_i\in\{E_1,\ldots,E_n\}$ and each successor of the root is in $\{E_1,\ldots,E_n\}\cup\mathbb{E}\subseteq\text{mod}(\psi_0\vee\ldots\vee\psi_m)$ by inductive hypothesis.
\end{proof}

\begin{proposition}\label{prop:lhsduality}
For every $\varphi\in\mono$, every positive example for $\varphi$ weakly simulates some example from $E^+_{\varphi}$.
\end{proposition}
\begin{proof}
By induction on formula complexity. Consider $\varphi\in\mono$ in normal form, i.e.~$\varphi$ is a disjunction of basic normal forms of level $d(\varphi)$. It suffices to show that the claim holds for all basic normal form disjuncts of $\varphi$, for if $M,s\models\varphi$ then $(M,s)$ satisfies some disjunct, and the positive examples for $\varphi$ are the union of all the positive examples for its disjuncts. 
\begin{itemize}
    \item[(i)] Let $\varphi=\pi$ be a conjunction of positive atoms and observe that $E^+_{\hat{\pi}}=\{\blacktriangledown_{\hat{\pi}}(\emptyloop)\}$. We look at 2 cases: (a) $R^M[s]=\emptyset$ or (b) $R^M[s]\ne\emptyset$. (a) We claim that the relation with a single pair $\{(r,s)\}$, where $r$ is the root of $\blacktriangledown_{\hat{\pi}}(\emptyloop)$, is a weak simulation $\blacktriangledown_{\hat{\pi}}(\emptyloop)\simu M,s$. [atom] $v(r)=\hat{\pi}$ and as $M,s\models\pi$ clearly $v(r)\subseteq v(s)$.
    [forth'] trivial since we can use the escape clause. [back'] trivial since $R^M[s]=\emptyset$. (b) Now we claim that $Z:=\{(r,s)\}\cup(\{r'\}\times dom(M))$ is a weak simulation, where $r'$ is the root of $\emptyloop$. It follows straight from the definition of weak simulation that each pair of the form $(r',t)$ satisfies all the clauses. Hence it suffices to show that all clauses hold of the root-pair $(r,s)\in Z$. [atom] $v(r)=\hat{\pi}$ and as $M,s\models\pi$ clearly $v(r)\subseteq v(s)$. [forth'] the only successor of $r$ is $\emptyloop$ so we can just use the escape clause. [back'] If $R^Mst$ then we clearly have the empty loopstate successor of $r$ with is $Z$-related to $t\in dom(M)$.
    \item[(ii)] Let $M,s\models\pi\wedge\Diamond\varphi_0\wedge\ldots\wedge\Diamond\varphi_n$. This means that for each $i\leq n$ there is a successor $t_i\in R^M[s]$ such that $M,t_i\models\varphi_i$. Therefore, by the inductive hypothesis, for each $i\leq n$ there is an example $(E_i,e_i)\in E^+_{\varphi_i}$ and a weak simulation $Z_i:E_i,e_i\simu M,t_i$, as well as a simulation $Z_{\pi}:\blacktriangledown_{\hat{\pi}}(\emptyloop)\simu M,s$. It follows that $\blacktriangledown_{\hat{\pi}}(\{(E_0,e_0),\ldots,(E_n,e_n)\})\in E^+_{\pi\wedge\Diamond\varphi_0\wedge\ldots\wedge\Diamond\varphi_n}$. The claim is that
    \[Z:=\bigcup_{1\leq i\leq n}Z_i\cup(\{\emptyloop\}\times dom(M))\cup\{(r,s)\}\]
    is a weak simulation $\blacktriangledown_{\hat{\pi}}(\{(E_1,e_1),\ldots,(E_n,e_n)\})\simu M,s$. We already know that all the $Z_i$ are weak simulations. Moreover, every pair of the form $(r',t)$ where $dom(\emptyloop)=\{r'\}$ satisfies all the clauses (cf. Lemma \ref{lem:Winitialemptyloop}). Hence we only have to show that the root-pair $(r,s)$ satisfies all the clauses. [atom] Follows from the [atom] clause of $Z_{\pi}$ (note that $\hat{\pi}$ may be empty here). [forth'] For the empty loopstate successor $r'$ we can use the escape clause, and every other successor of $r$ is of the form $e_i$ for some $i\leq n$. But then there is a matching successor $t_i$ of $s$ with $(e_i,u_i)\in Z_i\subseteq Z$. [back'] We can match every successor $u$ of $t$ with $r'$ since $(r',u)\in Z$ for all $u\in dom(M)$.
    \item[(iii)] Let $M,s\models\pi\wedge\Box(\psi_0\vee\ldots\vee\psi_m)$. By inductive hypothesis, this means that for each $t\in R^M[s]$ there is some example $(E_t,e_t)\in\bigcup_{j\leq m}E^+_{\psi_j}$ with a weak simulation $Z_t:E_t,e_t\simu M,t$, as well as simulation $Z_{\pi}:\blacktriangledown_{\hat{\pi}}(\emptyloop)\simu M,s$. Let $\mathbb{E}:=\{(E_t,e_t)\mid t\in R^M[s]\}$, then $\blacktriangledown_{\hat{\pi}}(\mathbb{E})\in E^+_{\pi\wedge\Box(\psi_0\vee\ldots\vee\psi_m)}$. The claim is that 
    \[Z:=\bigcup_{t\in R^M[s]}Z_t\cup\{(r,s)\}\]
    is a weak simulation $Z:\blacktriangledown_{\hat{\pi}}(\mathbb{E})\simu M,s$. We already know each $Z_t$ is a weak simulation, so it suffices to show that $(r,s)$ satisfies all the clauses. [atom] By the [atom] clause of $Z_{\pi}$ (note that $\hat{\pi}$ may be empty here). [forth'] Every successor of $r$ is of the form $e_t$ and hence there is some successor $t$ of $s$ with $(e_t,t)\in Z_t\subseteq Z$. [back'] For every successor $t$ of $s$ we have by construction a successor $e_t$ of $r$ with $(e_t,t)\in Z_t\subseteq Z$.
    \item[(iv)] Let $M,s\models\pi\wedge\Diamond\varphi_0\wedge\ldots\wedge\Diamond\varphi_n\wedge\Box(\psi_0\vee\ldots\vee\psi_m)$. By inductive hypothesis, this means that for each $t\in R^M[s]$ there is some example $(E_t,e_t)\in\bigcup_{j\leq m}E^+_{\psi_j}$ with a weak simulation $Z_t:E_t,e_t\simu M,t$. Moreover, $R^M[s]\ne\emptyset$ as e.g.  $M,s\models\Diamond\varphi_0$ and $\varphi_0\not\equiv\bot$. Set $\mathbb{E}:=\{(E_t,e_t)\mid t\in R^M[s]\}$ and note that $\mathbb{E}\ne\emptyset$ as $R^M[s]\ne\emptyset$. Moreover, by inductive hypothesis there is also a weak simulation $Z_{\pi}:\blacktriangledown_{\hat{\pi}}(\emptyloop)\simu M,s$ and for each $i\leq n$ there are examples $(E_i,e_i)\in\bigcup_{j\leq m}E^+_{\varphi_i\wedge\psi_j}$ such that there is a simulation $Z_i:E_i,e_i\simu M,t_i$ where $t_i$ is some successor of $s$. Set $\mathbb{E'}:=\{(E_i,e_i)\mid 0<i\leq n\}$. It follows that $\blacktriangledown_{\hat{\pi}}(\mathbb{E}\cup\mathbb{E'})\in E^+_{\pi\wedge\Diamond\varphi_0\wedge\ldots\wedge\Diamond\varphi_n\wedge\Box(\psi_0\vee\ldots\vee\psi_m)}$. We claim that the relation
    \[Z:=\bigcup_{i\leq n}Z_i\cup\bigcup_{t\in R^M[s]}Z_t\cup\{(r,s)\}\]
    is a weak simulation $Z:\blacktriangledown_{\hat{\pi}}(\mathbb{E}\cup\mathbb{E'})\simu M,s$. Again, since we already know each $Z_i$ and $Z_t$ to be weak simulations, it suffices to show that the root-pair $(r,s)$ satisfies all the clauses. [atom] By the [atom] clause of $Z_{\pi}$ (note that $\hat{\pi}$ may be empty here). [forth'] Every successor of the root is either  of the form $e_t$ or (ii) of the form $e_i$. In the former case we get a matching successor $t$ of $s$ with $(e_t,t)\in Z_t\subseteq Z$, and in the latter case we get a matching successor $t_i$ of $s$ with $(e_i,t_i)\in Z_i\subseteq Z$. [back'] every successor $t$ of $s$ has a matching successor $e_t$ of $r$ with $(e_t,t)\in Z_t\subseteq Z$ by construction.
\end{itemize}
\end{proof}

As it turns out, we can define the negative examples for formulas in $\mono$ in terms of the positive ones via the substitution operator $(\cdot)^{\neg}$ we saw in Sections~\ref{sec:category} and Section \ref{chap:fragments} when dualizing the positive result for \rrevision{$\mathcal{L}_{\Diamond,\wedge,\vee,\top,\bot}$} from \cite{tencatedalmauCQ} to \rrevision{$\mathcal{L}_{\Box,\wedge,\vee,\top,\bot}$}. We will do this in terms of the syntactic operator $\triangleleft$ defined below. Intuitively, $\triangleleft(\varphi)$ is the formula $\varphi$ with all logical symbols in it 'dualized' while it does nothing on propositional variables. That is;
\label{dualoperation}
\begin{align*}
    \triangleleft(\Box\varphi)=&\Diamond\triangleleft(\varphi)\\
    \triangleleft(\Diamond\varphi)=&\Box\triangleleft(\varphi)\\
    \triangleleft(\varphi\wedge\psi)=&\triangleleft(\varphi)\vee\triangleleft(\psi)\\
    \triangleleft(\varphi\vee\psi)=&\triangleleft(\varphi)\wedge\triangleleft(\psi)\\
    \triangleleft(p)=&p
\end{align*}
First of all, note that $\varphi\in\mono$ implies that $\triangleleft(\varphi)\in\mono$ because $\triangleleft(p)=p$. Moreover, observe that $\neg\varphi\equiv\triangleleft(\varphi)^{\neg}$. To see this, note that rewriting $\neg\varphi$ into negation normal form is precisely the same procedure as first dualizing all operators with $\triangleleft$ and then substituting $[\neg p/p]$ wherever possible with $(\cdot)^{\neg}$, modulo elimination of double negations in front of atoms.

\begin{definition}\label{def:characterisations}
For each $\varphi\in\mono$ set $\mathbb{E}_{\varphi}=(E^+_{\varphi}, E^-_{\varphi})$ where
\[E^+_{\varphi}:=E^+_{nf(\varphi)}\qquad\qquad\qquad\qquad E^-_{\varphi}:=\{(E^{\neg},e)\mid (E,e)\in E^+_{nf(\triangleleft(\varphi))}\}\]
\end{definition}

\begin{proposition}\label{prop:dualityvianegexamples}
For each $\varphi\in\mono$, $(E^+_{\varphi},E^-_{\varphi})$ is a weak simulation duality that $\varphi$ fits.
\end{proposition}
\begin{proof}
It remains to show that $\varphi$ fits $E^-_{\varphi}$ \rrevision{as negative examples} and that every negative example for $\varphi$ is weakly simulated by some negative example for $\varphi$ in $E^-_{\varphi}$. Let $(E',e')\in E^-_{\varphi}$. Then $(E',e')=(E,e)^{\neg}=(E^{\neg},e)$ for some $(E^{\neg},e)\in E^+_{nf(\triangleleft(\varphi))}$. So by Proposition~\ref{prop:fitsposexamples} $E^{\neg},e\models\triangleleft(\varphi)$ and by Proposition~\ref{prop:coherencenegoperations} $E,e\models\triangleleft(\varphi)^{\neg}$. But we have seen that that just means that $E,e\not\models\varphi$ since $\neg\varphi\equiv\triangleleft(\varphi)^{\neg}$. It follows that $\varphi$ fits $E^-_{\varphi}$ \rrevision{as negative examples}.

Finally, let $E,e\not\models\varphi$ be any negative example for $\varphi$. Since $\neg\varphi\equiv\triangleleft(\varphi)^{\neg}$, it follows that $E,e\models\triangleleft(\varphi)^{\neg}$. By Proposition~\ref{prop:coherencenegoperations} we have $E^{\neg},e\models\triangleleft(\varphi)$, and then by Proposition~\ref{prop:lhsduality} there is some $(E',e')\in E^+_{\triangleleft(\varphi)}$ such that $E',e'\simu E^{\neg},e$, whence by Proposition \ref{prop:flippingsimus} and idempotence of $(\cdot)^{\neg}$ also $E,e\simu (E')^{\neg},e'$. But if $(E',e')\in E^+_{\triangleleft(\varphi)}$ then $(E',e')^{\neg}=((E')^{\neg},e)$ is in $E^-_{\varphi}$. It follows that $(E,e)\in\downclosure E^-_{\varphi}$, so we have shown that $\downclosure E^-_{\varphi}=\text{finmod}(\neg\varphi)$. Together with Proposition~\ref{prop:lhsduality}, which shows that $\upclosure E^+_{\varphi}= \text{finmod}(\varphi)$, we conclude that $(E^+_{\varphi},E^-_{\varphi})$ is a weak simulation duality.
\end{proof}

We are now ready to prove Theorem~\ref{thm:main-mono}, restated here for convenience:

\thmmainmono*

\begin{proof}
Immediate from Proposition~\ref{prop:dualityvianegexamples} and Theorem~\ref{thm:dualitiesarecharacterisations}. The procedure computes first $nf(\varphi)$, $nf(\triangleleft(\varphi))$ and uses Definition \ref{def:posexamples} and the flipping operation $(\cdot)^{\neg}$ to generate $E^+_{\varphi}$ and $E^-_{\varphi}$.
\end{proof}

\begin{remark}
Theorem~\ref{thm:wsim-characterization} showed that the 
weak-simulation-preserved fragment of first-order logic 
is the fragment $\mono$ extended with $\top$ and $\bot$. 
This raises the question whether Theorem~\ref{thm:main-mono}
holds also for $\mono$ extended with $\top$ and $\bot$.
This is indeed the case:
it follows from Lemma~\ref{lem:emptyloopsatall} that a single pointed model suffices to characterise $\top$ and $\bot$ as individual formulas. We can use the fact that $\top$ does not have negative examples and $\bot$ does not have positive examples, and we can set $E^+_{\top}:=\{\emptyloop\}$ and $E^-_{\bot}:=\{\fullloop\}$.
\end{remark}

\section{Further Results}
\label{sec:further}
In this section, we discuss complexity issues and additional extensions of the modal languages considered that admit finite characterisations. We will begin with discussing the complexity of computing finite characterisations.

\subsection{Size bounds and computational complexity of computing finite characterizations}
\label{sec:complexity}

Recall Theorem \ref{thm:Balder}, which states that $\mathbb{C}(\mathcal{L}_{\Diamond,\wedge,\vee,\top,\bot}[\Prop])$ admits finite characterizations, and whose proof relies on results from~\cite{AlexeCKT2011,tencatedalmauCQ} concerning finite
characterizations for c-acyclic unions of conjunctive queries.
In fact, the latter
results imply that $\mathbb{C}(\posexcon[\Prop])$ admits polynomial-time computable finite characterisations and that $\mathbb{C}(\mathcal{L}_{\Diamond,\wedge,\vee,\top}[\Prop])$ admits exponential-time computable finite characterisations. 
Inspecting our proof of Theorem \ref{thm:Balder}, we see that
this extends also to the corresponding fragments that include
$\bot$. Therefore, $\mathbb{C}(\mathcal{L}_{\Diamond,\wedge,\top,\bot}[\Prop])$ admits polynomial-time computable finite characterisations and that $\mathbb{C}(\mathcal{L}_{\Diamond,\wedge,\vee,\top,\bot}[\Prop])$ admits exponential-time computable characterisations. 

Our construction of finite characterizations for $\mono$-formulas, on the other hand, is non-elementary. Indeed, 
we will now give a non-elementary lower bound on the size of characterisations w.r.t. $\mono$, by showing that any characterisation of $\Box^np$ w.r.t. $\mono[\{p\}]$ must contain at least non-elementarily many positive examples. We will use the following notation for non-elementary functions:
\begin{align*}
    \text{tower}(1,m)&\;:=\;m\\
    \text{tower}(n+1,m)&\;:=\;2^{\text{tower}(n,m)}
\end{align*}
Thus, a non-elementary function $f(n)$ is one that cannot be bounded by any $\text{tower}(k,\text{poly}(n))$ for any fixed $k$ and polynomial function $\text{poly}(n)$. 

We will use the following general strategy for proving the lower bound. Namely, we show that our characterisations are `minimal' in the sense that no proper subset forms a characterisation. Using the fact that our characterisations actually form weak simulation dualities in $\mathbf{wSim}[\{p\}]$, it then follows that the size of the characterisations for individual formulas are actually lower bounds on the size on any such characterisation.

\begin{lemma}\label{lem:minimality}
For each $n$, $E^{+}_{\Box^np}$ is `minimal' in the sense that for all proper subsets $E'\subset E^{+}_{\Box^np}$ we have $\upclosure E'\ne\textup{finmod}(\Box^np)$.
\end{lemma}
\begin{proof}
We show this by induction on $n$. The base case $n=0$ is trivial since $\Box^op=p$ and $|E^+_{p}|=1$ has only a single example and clearly $\upclosure\emptyset$ is empty so is not exactly $\text{finmod}(\Box^np)$. Now suppose the claim holds for $E^+_{\Box^np}$ and consider any $\blacktriangledown_{\emptyset}(\mathbb{E})\in E^+_{\Box^{n+1}p}$ where $\mathbb{E}\subseteq E^+_{\Box^np}$. We want to construct a formula in $\mono$ which is only refuted on this particular positive example in $E^+_{\Box^{n+1}p}$, but \rrevision{not on other examples} in $E^+_{\Box^{n+1}p}$. It follows from the fact that $E^+_{\Box^np}$ is minimal that for every positive example $(E,e)\in\mathbb{E}\subseteq E^+_{\Box^np}$ there is some formula $\psi_{(E,e)}\in\mono$ such that $\psi_{(E,e)}$ is only refuted on $(E,e)$, but true on all other examples in $E^+_{\Box^np}$. Define: \footnote{Interestingly, this is precisely the dual of the modal nabla operator (a.k.a. the `cover modality') $\nabla$ \cite{MOSS1999277,janin1995}. }
\[\psi_{\blacktriangledown_{\emptyset}(\mathbb{E})}:=(\bigvee_{(E,e)\in\mathbb{E}}\Box\;\psi_{(E,e)})\vee\Diamond(\bigwedge_{(E,e)\in\mathbb{E}}\psi_{(E,e)})\]
We claim that $\blacktriangledown_{\emptyset}(\mathbb{E})$ is the only example from $E^+_{\Box^{n+1}p}$ that refutes $\psi_{\blacktriangledown_{\emptyset}(\mathbb{E})}$. For take any other $\nabla_{\emptyset}(\mathbb{E'})\in E^+_{\Box^{n+1}p}$ with $\mathbb{E}\ne\mathbb{E}'$. Then either (a) $\mathbb{E}\not\subseteq\mathbb{E'}$ or (b) $\mathbb{E'}\not\subseteq\mathbb{E}$. If (a) then there is some $(E,e)\in\mathbb{E}$ such that $(E,e)\not\in\mathbb{E'}$. But then by the inductive hypothesis there is some $\psi_{(E,e)}\in\mono$ that is refuted at $(E,e)$ but satisfied at all other examples in $E^+_{\Box^np}$, including all examples in $\mathbb{E'}$. It follows that $\nabla_{\emptyset}(\mathbb{E'})\models\Box\psi_{(E,e)}$, which is a disjunct of $\psi_{\blacktriangledown_{\emptyset}(\mathbb{E})}$. Moreover, if (b) then there is some $(E',e')\in\mathbb{E'}$ such that $(E',e')\not\in\mathbb{E}$. \revision{By the inductive hypothesis, for each $(E,e)\in\mathbb{E}$ there is a formula $\psi_{(E,e)}\in\mono$ that is only refuted on $(E,e)$ but \rrevision{not on other examples} in $E^+_{\Box^np}$. It follows that $E',e'\models
\bigwedge_{(E,e)\in\mathbb{E}}\psi_{(E,e)}$} and so since $(E',e')\in\mathbb{E'}$ we get $\nabla_{\emptyset}(\mathbb{E'})\models\Diamond(\bigwedge_{(E,e)\in\mathbb{E}}\psi_{(E,e)})$, which is a disjunct of $\psi_{\blacktriangledown_{\emptyset}(\mathbb{E})}$. 
\end{proof}

\begin{theorem}
Every finite characterisation of $\Box^np$ w.r.t. $\mathbb{C}(\mono[\{p\}])$ has at least $\text{tower}(n,2)$ many positive examples. Dually, every finite characterisation of $\Diamond^np$ w.r.t. $\mathbb{C}(\mono[\{p\}])$ has at least $\text{tower}(n,2)$ many negative examples.
\end{theorem}
\begin{proof}
Suppose for contradiction that $\mathbb{E}=(E^+,E^-)$ is a characterisation of $\Box^np$ w.r.t. $\mono$ with $m=|E^+|<|E^+_{\Box^np}|=\text{tower}(n,2)$ and let us we write $E^+=\{(E_1,e_1),\ldots,(E_m,e_m)\}$. We have shown that $\mathbb{E}_{\Box^np}$ characterises $\Box^np$ w.r.t. $\mono$ and is moreover a weak simulation duality. \rrevision{Hence, it} follows that for each $(E_i,e_i)\in E^+$ with $1\leq i\leq m$ there is some $(E'_i,e'_i)\in E^+_{\Box^np}$ such that $(E'_i,e'_i)\simu(E_i,e_i)$. But as $m< \text{tower}(n,2)=|E^+_{\Box^np}|$, there is some \textit{proper} subset $\mathbb{E}\subset E^+_{\Box^np}$ of size $m$, i.e. $\mathbb{E}=\{(E'_1,e'_1),\ldots,(E'_m,e'_m)\}$) such that $E^+\subseteq\upclosure\mathbb{E}$. But by Lemma \ref{lem:minimality} we know that $\mathbb{E}_{\Box^np}$ is minimal so for any $(E,e)\in E^+_{\Box^np}-\mathbb{E}$ (which is non-empty by assumption), there must be some $\psi_{(E,e)}$ such that $E,e\not\models\psi_{(E,e)}$ but all other examples in $E^+_{\Box^np}$ satisfy $\psi_{(E,e)}$. In particular, each example in $\mathbb{E}$ satisfies $\psi_{(E,e)}$ and hence by preservation each example in $E^+$ does as well as $E^+\subseteq\upclosure\mathbb{E}$. But then $\Box^np\wedge\psi_{(E,e)}$ fits $\mathbb{E}$ yet $\Box^np\not\equiv\Box^np\wedge\psi_{(E,e)}$ as witnessed by $E,e\models\Box^np\wedge\neg\psi_{(E,e)}$, contradicting our assumption that $\mathbb{E}$ characterises $\Box^{n}p$ w.r.t. $\mono$.
\end{proof}

\subsection{Uniform Formulas}

In Section~\ref{chap:fragments}, we showed that $\mathbb{C}(\mathcal{L}_{\Box,\Diamond,\wedge,\bot}[\Prop])$ and $\mathbb{C}(\mathcal{L}_{\Box,\Diamond,\vee,\top}[\Prop])$ do not admit finite characterisations, even for $\Prop=\emptyset$. Specifically, $\bot$ is not finitely characterisable in the former and $\top$ is not finitely characterisable in the latter. It follows that  $\mathbb{C}(\mathcal{L}_{\Box,\Diamond,\wedge,\neg_{\text{at}}}[\Prop])$, for nonempty $\Prop$, does not admit finite characterisations either, because $\bot$ is definable as $p\wedge\neg_{\text{at}}p$, and likewise for the dual class $\mathbb{C}(\mathcal{L}_{\Box,\Diamond,\vee,\neg_{\text{at}}}[\Prop])$. In other words, our positive result for
$\mathbb{C}(\mono)$ crucially relies on the absence of negation.

We will now show that 
our positive result for $\mathbb{C}(\mono[\Prop])$ extends to a larger class of \emph{uniform} $\mathcal{L}_{\Diamond,\Box,\land,\lor,\neg_{\text{at}}}$-formulas, which allow atomic negation only in a very restricted way. For any $P,Q$ disjoint finite sets of propositional variables, and a modal fragment $\mathcal{L}_S$, let $\mathcal{L}_S[P;Q]$ denote the set of $\mathcal{L}$-formulas that only use unnegated occurrences of variables $p\in P$ and only negated occurrences $\neg_{at}q$ for variables $q\in Q$. 

We will show that $\mathbb{C}(\mathcal{L}_{\Box,\Diamond,\wedge,\vee,\neg_{\text{at}}}[P;Q])$ admits finite characterisations, through a suitable reduction to  
$\mathbb{C}(\mathcal{L}_{\Box,\Diamond,\wedge,\vee}[\Prop])$
The idea is that we can treat negated variables from $Q$ \emph{as if they were positive}. For each $q\in Q$ we create fresh variable $q^{\neg}$ that stands for the negation $\neg q$ of $q$. Let $Q^{\neg}:=\{q^{\neg}\mid q\in Q\}$. 

\begin{definition}
Given a pointed model $(M,s)$, let $(M,s)^{\neg Q}:=(M^{\neg},s)$ where $M^{\neg Q}:=(dom(M),R,v^{\neg Q})$ is obtained from $M$ by defining $v^{\neg Q}:dom(M)\to\mathcal{P}(P\cup Q^{\neg})$ as $v^{\neg Q}(t):=v(t)\cup\{q^{\neg}\mid q\in Q\;\text{and}\;q\not\in v(t)\}$ for all $t\in dom(M)$.  Next, given a modal formula $\varphi\in\mathcal{L}_{\text{full}}[P;Q]$ where $Q=\{q_1,\ldots,q_n\}$, let 
\[\varphi^{\neg Q}:=\varphi[q_1^{\neg}/\neg q_1]\;\ldots\;[q_n^{\neg}/\neg q_n]\]
Note that if $\varphi\in\mathcal{L}_{\text{full}}[P;Q]$ then $\varphi^{\neg Q}\in\mathcal{L}_{\text{full}}[P\cup Q^{\neg}]$.
\end{definition}

\begin{lemma}\label{lem:flippinrelativised}
For all modal formulas $\varphi,\psi\in\mathcal{L}_{\text{full}}[P;Q]$ and pointed models $(M,s)$;
\[M,s\models\varphi\;\;\textrm{iff}\;\; M^{\neg Q},s\models\varphi^{\neg Q}\qquad\text{and}\qquad\varphi\equiv\psi\;\;\text{iff}\;\;\varphi^{\neg Q}\equiv\psi^{\neg Q}\]
\end{lemma}
\begin{proof}
Only the atomic case differs from Proposition~\ref{prop:coherencenegoperations} and for any state $s$ and $q\in Q$ we have $s\models\neg_{\text{at}}q$ iff $q\not\in v(s)$ iff $q^{\neg}\in v^{\neg Q}(s)$ iff $s\models q^{\neg}$. Next, observe that $\varphi\equiv\psi$ iff $\varphi\leftrightarrow\psi$ is valid. Furthermore, $\varphi\leftrightarrow\psi$ is valid iff $(\varphi\leftrightarrow\psi)^{\neg Q}=\varphi^{\neg Q}\leftrightarrow\psi^{\neg Q}$ is valid because validity is closed under uniform substitution \cite{bluebook}, which of course holds iff $\varphi^{\neg Q}\equiv\psi^{\neg Q}$.
\end{proof}

\begin{theorem}
For all disjoint finite sets $P,Q\subseteq\Prop$, the concept class $\mathbb{C}(\mathcal{L}_{\Box,\Diamond,\wedge,\vee,\neg_{\text{at}}}[P;Q])$ admits finite characterisations.
\end{theorem}
\begin{proof}
Take any $\varphi\in\mathcal{L}_{\Box,\Diamond,\wedge,\vee,\neg_{\text{at}}}[P;Q]$ and consider $\varphi^{\neg Q}\in\mathcal{L}_{\Box,\Diamond,\wedge,\vee}[P\cup Q^{\neg}]$. By Theorem~\ref{thm:main-mono}, $\varphi^{\neg Q}$ has a finite characterisation $(E^+,E^-)$ with respect to $\mathcal{L}_{\Box,\Diamond,\wedge,\vee}[P\cup Q^{\neg}]$. Define
\[E^+_{\varphi}:=\{(E,e)^{\neg Q}\mid (E,e)\in E^+_{\varphi^{\neg Q}}\}\qquad E^-_{\varphi}:=\{(E,e)^{\neg Q}\mid (E,e)\in E^-_{\varphi^{\neg Q}}\}\]
and consider any $\psi\in\mathcal{L}_{\Box,\Diamond,\wedge,\vee}[P;Q]$ that fits $(E^+_{\varphi}, E^-_{\varphi})$. By Theorem \ref{lem:flippinrelativised} it follows that $\psi^{\neg Q}$ fits $(E^+_{\varphi^{\neg Q}},E^-_{\varphi^{\neg Q}})$ and hence $\varphi^{\neg Q}\equiv\psi^{\neg Q}$ by the fact that this is a characterisation of $\varphi^{\neg Q}$. But then by lemma \ref{lem:flippinrelativised} we get that $\varphi\equiv\psi$.
\end{proof}





\subsection{Poly-modal formulas}
So far we have only looked at modal languages with a single modality $\Diamond$ or with both the $\Diamond$ modality and its dual $\Box$, and the corresponding Kripke models with only \emph{one} accessibility relation. However, our results generalise straightforwardly to their respective poly-modal generalisations. This is relevant from the point of
view of potential applications in description logic.
Our proofs extend ``pointwise'' to the poly-modal setting. 
In particular, one obtains generalised versions of Fine normal form, where basic normal forms are of the form
\[\pi\wedge\bigwedge_{i}\Box_i\varphi_i\wedge\Diamond_i\varphi_{i,0}\wedge\ldots\wedge\Diamond\varphi_{i,n_i}\]
Further, the definition of weak simulations and our category-theoretic analysis of $\mathbf{wSim}[\Prop]$ extend straightforwardly to the poly-modal setting. The relevant construction of positive examples for poly-modal formulas in basic normal form proceeds by recursively computing the characterisations for individual component basic normal forms with one modality and merging their roots. Again, negative examples can be obtained from the positive ones via the endo-functor $(\cdot)^\neg$.

\section{Conclusion}
In this paper, we investigated under what conditions modal formulas admit finite characterisations. First, we focused on the full modal language and showed that even over restricted classes of frames it does not admit finite characterisations, except in trivial cases. This motivated our further investigation into modal fragments, where we fixed the frame class parameter to the class of all (finite) frames, but 
we restrict the set of connectives. Here, we obtain a number
of positive and negative results. In particular, in section \ref{sec:mono} we proved our main result that $\mono$ admits finite characterisations. Finally, in section \ref{sec:further} we proved tight non-elementary bounds on the size of characterisations for $\mono$, and extended our main result to classes of uniform formulas and indicated how poly-modal generalisations could be obtained.

In future work, we would like to close the question on whether $\mathbb{C}(\mathcal{L}_{\Box,\Diamond,\wedge,\top})$ and $\mathbb{C}(\mathcal{L}_{\Diamond,\wedge,\vee,\neg_{\text{at}}})$ admit finite characterisations. Furthermore, we conjecture that $\mathbb{C}(\mathcal{L}_{\Box,\Diamond,\wedge})$ and even $\mathbb{C}(\mathcal{L}_{\Box,\Diamond,\wedge,\top})$ admits polynomial-sized characterisations, and might even support a polynomial-time exact learning algorithm with membership queries. 

Finally, we leave the question of the existence of characterisations for fragments w.r.t. restricted frame classes wide open. Recent results in~\cite{cate2023rightadjoints} may provide an initial starting point here.

\begin{acks}
Koudijs is supported by the NFR project ``Learning Description Logic Ontologies'', grant number 316022. Ten Cate is supported by the European Union’s Horizon 2020 research and innovation programme under grant MSCA-101031081. We are grateful to Nick Bezhanishvili, Malvin Gattinger and Frank Wolter for comments on earlier versions of this material. We also thank the anonymous reviewers for their detailed and helpful feedback.
\end{acks}

\bibliographystyle{ACM-Reference-Format}
\bibliography{bibfile}


\begin{thebibliography}{20}


\ifx \showCODEN    \undefined \def \showCODEN     #1{\unskip}     \fi
\ifx \showDOI      \undefined \def \showDOI       #1{#1}\fi
\ifx \showISBNx    \undefined \def \showISBNx     #1{\unskip}     \fi
\ifx \showISBNxiii \undefined \def \showISBNxiii  #1{\unskip}     \fi
\ifx \showISSN     \undefined \def \showISSN      #1{\unskip}     \fi
\ifx \showLCCN     \undefined \def \showLCCN      #1{\unskip}     \fi
\ifx \shownote     \undefined \def \shownote      #1{#1}          \fi
\ifx \showarticletitle \undefined \def \showarticletitle #1{#1}   \fi
\ifx \showURL      \undefined \def \showURL       {\relax}        \fi
\providecommand\bibfield[2]{#2}
\providecommand\bibinfo[2]{#2}
\providecommand\natexlab[1]{#1}
\providecommand\showeprint[2][]{arXiv:#2}

\bibitem[\protect\citeauthoryear{Alexe, Cate, Kolaitis, and Tan}{Alexe
  et~al\mbox{.}}{2011}]%
        {AlexeCKT2011}
\bibfield{author}{\bibinfo{person}{Bogdan Alexe}, \bibinfo{person}{Balder~ten
  Cate}, \bibinfo{person}{Phokion~G. Kolaitis}, {and}
  \bibinfo{person}{Wang-Chiew Tan}.} \bibinfo{year}{2011}\natexlab{}.
\newblock \showarticletitle{Characterizing Schema Mappings via Data Examples}.
\newblock \bibinfo{journal}{\emph{ACM Trans. Database Syst.}}
  \bibinfo{volume}{36}, \bibinfo{number}{4}, Article \bibinfo{articleno}{23}
  (\bibinfo{date}{Dec.} \bibinfo{year}{2011}), \bibinfo{numpages}{48}~pages.
\newblock
\showISSN{0362-5915}
\urldef\tempurl%
\url{https://doi.org/10.1145/2043652.2043656}
\showDOI{\tempurl}


\bibitem[\protect\citeauthoryear{Angluin}{Angluin}{1988}]%
        {AngluinQueriesConcepts}
\bibfield{author}{\bibinfo{person}{Dana Angluin}.}
  \bibinfo{year}{1988}\natexlab{}.
\newblock \showarticletitle{Queries and Concept Learning}.
\newblock \bibinfo{journal}{\emph{Mach. Learn.}} \bibinfo{volume}{2},
  \bibinfo{number}{4} (\bibinfo{year}{1988}), \bibinfo{pages}{319–342}.
\newblock
\showISSN{0885-6125}
\urldef\tempurl%
\url{https://doi.org/10.1023/A:1022821128753}
\showDOI{\tempurl}


\bibitem[\protect\citeauthoryear{Blackburn, de~Rijke, and Venema}{Blackburn
  et~al\mbox{.}}{2001}]%
        {bluebook}
\bibfield{author}{\bibinfo{person}{Patrick Blackburn}, \bibinfo{person}{Maarten
  de Rijke}, {and} \bibinfo{person}{Yde Venema}.}
  \bibinfo{year}{2001}\natexlab{}.
\newblock \bibinfo{booktitle}{\emph{Modal Logic}}. \bibinfo{series}{Cambridge
  Tracts in Theoretical Computer Science}, Vol.~\bibinfo{volume}{53}.
\newblock \bibinfo{publisher}{Cambridge University Press},
  \bibinfo{address}{Cambridge}.
\newblock


\bibitem[\protect\citeauthoryear{Chang and Keisler}{Chang and Keisler}{1973}]%
        {ChangKeisler}
\bibfield{author}{\bibinfo{person}{Chen~Chung Chang} {and}
  \bibinfo{person}{H.~Jerome Keisler}.} \bibinfo{year}{1973}\natexlab{}.
\newblock \bibinfo{booktitle}{\emph{Model theory}}.
\newblock \bibinfo{publisher}{North-Holland Pub. Co.; American Elsevier
  Amsterdam, New York}.
\newblock


\bibitem[\protect\citeauthoryear{Fine}{Fine}{1975}]%
        {Finenormalform}
\bibfield{author}{\bibinfo{person}{Kit Fine}.} \bibinfo{year}{1975}\natexlab{}.
\newblock \showarticletitle{Normal Forms in Modal Logic}.
\newblock \bibinfo{journal}{\emph{Notre Dame Journal of Formal Logic}}
  \bibinfo{volume}{16}, \bibinfo{number}{2} (\bibinfo{year}{1975}),
  \bibinfo{pages}{229--237}.
\newblock
\urldef\tempurl%
\url{https://doi.org/10.1305/ndjfl/1093891703}
\showDOI{\tempurl}


\bibitem[\protect\citeauthoryear{Fortin, Konev, Ryzhikov, Savateev, Wolter, and
  Zakharyaschev}{Fortin et~al\mbox{.}}{2022}]%
        {LTLWolter}
\bibfield{author}{\bibinfo{person}{Marie Fortin}, \bibinfo{person}{Boris
  Konev}, \bibinfo{person}{Vladislav Ryzhikov}, \bibinfo{person}{Yury
  Savateev}, \bibinfo{person}{Frank Wolter}, {and} \bibinfo{person}{Michael
  Zakharyaschev}.} \bibinfo{year}{2022}\natexlab{}.
\newblock \showarticletitle{Unique Characterisability and Learnability of
  Temporal Instance Queries}. In \bibinfo{booktitle}{\emph{{Proceedings of KR
  2022}}}. \bibinfo{pages}{163--173}.
\newblock


\bibitem[\protect\citeauthoryear{Funk, Jung, and Lutz}{Funk
  et~al\mbox{.}}{2021}]%
        {Funk2021ELr}
\bibfield{author}{\bibinfo{person}{Maurice Funk},
  \bibinfo{person}{Jean~Christoph Jung}, {and} \bibinfo{person}{Carsten Lutz}.}
  \bibinfo{year}{2021}\natexlab{}.
\newblock \showarticletitle{Actively Learning Concepts and Conjunctive Queries
  under ELr-Ontologies}. In \bibinfo{booktitle}{\emph{Proceedings of IJCAI
  2021}}. \bibinfo{pages}{1887--1893}.
\newblock


\bibitem[\protect\citeauthoryear{Funk, Jung, and Lutz}{Funk
  et~al\mbox{.}}{2022}]%
        {FunkJL2022}
\bibfield{author}{\bibinfo{person}{Maurice Funk},
  \bibinfo{person}{Jean~Christoph Jung}, {and} \bibinfo{person}{Carsten Lutz}.}
  \bibinfo{year}{2022}\natexlab{}.
\newblock \showarticletitle{Frontiers and Exact Learning of {ELI} Queries under
  DL-Lite Ontologies}. In \bibinfo{booktitle}{\emph{Proceedings of the
  Thirty-First International Joint Conference on Artificial Intelligence,
  {IJCAI} 2022, Vienna, Austria, 23-29 July 2022}},
  \bibfield{editor}{\bibinfo{person}{Luc~De Raedt}} (Ed.).
  \bibinfo{publisher}{ijcai.org}, \bibinfo{pages}{2627--2633}.
\newblock
\urldef\tempurl%
\url{https://doi.org/10.24963/ijcai.2022/364}
\showDOI{\tempurl}


\bibitem[\protect\citeauthoryear{Janin and Walukiewicz}{Janin and
  Walukiewicz}{1995}]%
        {janin1995}
\bibfield{author}{\bibinfo{person}{D. Janin} {and} \bibinfo{person}{I.
  Walukiewicz}.} \bibinfo{year}{1995}\natexlab{}.
\newblock \showarticletitle{Automata for the Modal $\mu$-calculus and Related
  Results}. In \bibinfo{booktitle}{\emph{Proceedings of the Twentieth
  International Symposium on Mathematical Foundations of Computer Science,
  MFCS'95}} \emph{(\bibinfo{series}{LNCS})}, Vol.~\bibinfo{volume}{969}.
  \bibinfo{publisher}{Springer}, \bibinfo{pages}{552--562}.
\newblock


\bibitem[\protect\citeauthoryear{Kerdiles}{Kerdiles}{2001}]%
        {kerdiles}
\bibfield{author}{\bibinfo{person}{Gwen Kerdiles}.}
  \bibinfo{year}{2001}\natexlab{}.
\newblock \bibinfo{booktitle}{\emph{Saying it with Pictures: a Logical
  Landscape of Conceptual Graphs}}.
\newblock \bibinfo{publisher}{ILLC Dissertation Series (DS)}.
\newblock


\bibitem[\protect\citeauthoryear{Koudijs}{Koudijs}{2022}]%
        {MyThesis}
\bibfield{author}{\bibinfo{person}{Raoul Koudijs}.}
  \bibinfo{year}{2022}\natexlab{}.
\newblock \bibinfo{title}{Learning Modal Formulas via Dualities}.
\newblock
\newblock
\urldef\tempurl%
\url{https://eprints.illc.uva.nl/id/eprint/1957/1/MoL-2022-07.text.pdf}
\showURL{%
\tempurl}
\newblock
\shownote{ILLC MoL Series (MSc Thesis).}


\bibitem[\protect\citeauthoryear{Kurtonina and Rijke}{Kurtonina and
  Rijke}{1997}]%
        {KurtoninaDeRijke}
\bibfield{author}{\bibinfo{person}{Natasha Kurtonina} {and}
  \bibinfo{person}{Maarten~De Rijke}.} \bibinfo{year}{1997}\natexlab{}.
\newblock \showarticletitle{Simulating without Negation}.
\newblock \bibinfo{journal}{\emph{Journal of Logic and Computation}}
  \bibinfo{volume}{7}, \bibinfo{number}{4} (\bibinfo{date}{08}
  \bibinfo{year}{1997}), \bibinfo{pages}{501--522}.
\newblock


\bibitem[\protect\citeauthoryear{Mannila and R{\"a}ih{\"a}}{Mannila and
  R{\"a}ih{\"a}}{1986}]%
        {MannilaR86}
\bibfield{author}{\bibinfo{person}{Heikki Mannila} {and}
  \bibinfo{person}{Kari-Jouko R{\"a}ih{\"a}}.} \bibinfo{year}{1986}\natexlab{}.
\newblock \showarticletitle{Test Data for Relational Queries}. In
  \bibinfo{booktitle}{\emph{PODS}}. \bibinfo{pages}{217--223}.
\newblock


\bibitem[\protect\citeauthoryear{Moss}{Moss}{1999}]%
        {MOSS1999277}
\bibfield{author}{\bibinfo{person}{Lawrence~S. Moss}.}
  \bibinfo{year}{1999}\natexlab{}.
\newblock \showarticletitle{Coalgebraic logic}.
\newblock \bibinfo{journal}{\emph{Annals of Pure and Applied Logic}}
  \bibinfo{volume}{96}, \bibinfo{number}{1} (\bibinfo{year}{1999}),
  \bibinfo{pages}{277--317}.
\newblock
\showISSN{0168-0072}
\urldef\tempurl%
\url{https://doi.org/10.1016/S0168-0072(98)00042-6}
\showDOI{\tempurl}


\bibitem[\protect\citeauthoryear{Rosen}{Rosen}{1997}]%
        {Rosen1997:modal}
\bibfield{author}{\bibinfo{person}{Eric Rosen}.}
  \bibinfo{year}{1997}\natexlab{}.
\newblock \showarticletitle{Modal Logic over Finite Structures}.
\newblock \bibinfo{journal}{\emph{Journal of Logic, Language, and Information}}
  \bibinfo{volume}{6}, \bibinfo{number}{4} (\bibinfo{year}{1997}),
  \bibinfo{pages}{427--439}.
\newblock
\showISSN{09258531, 15729583}
\urldef\tempurl%
\url{http://www.jstor.org/stable/40181601}
\showURL{%
\tempurl}


\bibitem[\protect\citeauthoryear{Sestic}{Sestic}{2023}]%
        {PatrikThesis}
\bibfield{author}{\bibinfo{person}{Patrik Sestic}.}
  \bibinfo{year}{2023}\natexlab{}.
\newblock \bibinfo{title}{Unique Characterisability of Linear Temporal Logic}.
\newblock
\newblock
\newblock
\shownote{ILLC MoL Series (MSc Thesis), to appear.}


\bibitem[\protect\citeauthoryear{Staworko and Wieczorek}{Staworko and
  Wieczorek}{2015}]%
        {StaworkoW15}
\bibfield{author}{\bibinfo{person}{Slawek Staworko} {and}
  \bibinfo{person}{Piotr Wieczorek}.} \bibinfo{year}{2015}\natexlab{}.
\newblock \showarticletitle{Characterizing {XML} Twig Queries with Examples}.
  In \bibinfo{booktitle}{\emph{18th International Conference on Database
  Theory, {ICDT} 2015, March 23-27, 2015, Brussels, Belgium}}
  \emph{(\bibinfo{series}{LIPIcs})}, \bibfield{editor}{\bibinfo{person}{Marcelo
  Arenas} {and} \bibinfo{person}{Mart{\'{\i}}n Ugarte}} (Eds.),
  Vol.~\bibinfo{volume}{31}. \bibinfo{publisher}{Schloss Dagstuhl -
  Leibniz-Zentrum f{\"{u}}r Informatik}, \bibinfo{pages}{144--160}.
\newblock
\urldef\tempurl%
\url{https://doi.org/10.4230/LIPIcs.ICDT.2015.144}
\showDOI{\tempurl}


\bibitem[\protect\citeauthoryear{ten Cate and Dalmau}{ten Cate and
  Dalmau}{2022}]%
        {tencatedalmauCQ}
\bibfield{author}{\bibinfo{person}{Balder ten Cate} {and}
  \bibinfo{person}{Victor Dalmau}.} \bibinfo{year}{2022}\natexlab{}.
\newblock \showarticletitle{Conjunctive Queries: Unique Characterizations and
  Exact Learnability}.
\newblock \bibinfo{journal}{\emph{{ACM} Trans. Database Syst.}}
  \bibinfo{volume}{47}, \bibinfo{number}{4} (\bibinfo{year}{2022}),
  \bibinfo{pages}{14:1--14:41}.
\newblock
\urldef\tempurl%
\url{https://doi.org/10.1145/3559756}
\showDOI{\tempurl}


\bibitem[\protect\citeauthoryear{ten Cate, Dalmau, and Opršal}{ten Cate
  et~al\mbox{.}}{2023}]%
        {cate2023rightadjoints}
\bibfield{author}{\bibinfo{person}{Balder ten Cate}, \bibinfo{person}{Víctor
  Dalmau}, {and} \bibinfo{person}{Jakub Opršal}.}
  \bibinfo{year}{2023}\natexlab{}.
\newblock \bibinfo{title}{Right-Adjoints for Datalog Programs, and Homomorphism
  Dualities over Restricted Classes}.
\newblock
\newblock
\showeprint[arxiv]{2302.06366}~[cs.LO]


\bibitem[\protect\citeauthoryear{van Benthem}{van Benthem}{1976}]%
        {benthem1976}
\bibfield{author}{\bibinfo{person}{J. van Benthem}.}
  \bibinfo{year}{1976}\natexlab{}.
\newblock \emph{\bibinfo{title}{Modal Correspondence Theory}}.
\newblock \bibinfo{thesistype}{Ph.D. Dissertation}. \bibinfo{school}{University
  of Amsterdam}.
\newblock


\end{thebibliography}

\appendix

\end{document}